\documentclass[sigplan,10pt]{acmart}
\acmSubmissionID{52}

\ccsdesc[500]{Networks~Data center networks}
\ccsdesc[500]{Networks~Traffic engineering algorithms}
\usepackage[english]{babel}

\usepackage{tikz}
\usepackage{amsmath}
\usepackage{minted}
\usepackage{algorithm}
\usepackage{algpseudocode}
\usepackage{tabularx}
\usepackage{multirow}
\usepackage{caption}
\usepackage{svg}
\usepackage{subcaption}
\usepackage{enumitem}
\usepackage{xcolor}
\usepackage[colorinlistoftodos, prependcaption]{todonotes}

\newcommand{\tom}[1]{\textcolor{black}{#1}}

\newcommand{\changed}[1]{\textcolor{black}{#1}}

\copyrightyear{2026}
\acmYear{2026}
\setcopyright{cc}
\setcctype{by}
\acmConference[EUROSYS '26]{21st European Conference on Computer Systems}{April 27--30, 2026}{Edinburgh, Scotland Uk}
\acmBooktitle{21st European Conference on Computer Systems (EUROSYS '26), April 27--30, 2026, Edinburgh, Scotland Uk}
\acmPrice{}
\acmDOI{10.1145/3767295.3769320}
\acmISBN{979-8-4007-2212-7/26/04}

\begin{document}


\title{REPS: \underline{R}ecycled \underline{E}ntropy \underline{P}acket \underline{S}praying \\ for Adaptive Load Balancing and Failure Mitigation}

\author{Tommaso Bonato}
\affiliation{%
  \institution{ETH Z\"urich}
  \city{}
  \country{}
}
\affiliation{%
  \institution{Microsoft}
  \city{}
  \country{}
}

\author{Abdul Kabbani}
\affiliation{%
  \institution{Microsoft}
  \city{}
  \country{}
}

\author{Ahmad Ghalayini}
\affiliation{%
  \institution{Microsoft}
  \city{}
  \country{}
}

\author{Michael Papamichael}
\authornote{Work done while at Microsoft; the author has since left the company.}
\affiliation{%
  \institution{Microsoft}
  \city{}
  \country{}
}

\author{Mohammad Dohadwala}
\authornotemark[1] 
\affiliation{%
  \institution{Microsoft}
  \city{}
  \country{}
}

\author{Lukas Gianinazzi}
\affiliation{%
  \institution{ETH Z\"urich}
  \city{}
  \country{}
}

\author{Mikhail Khalilov}
\affiliation{%
  \institution{ETH Z\"urich}
  \city{}
  \country{}
}

\author{Elias Achermann}
\affiliation{%
  \institution{ETH Z\"urich}
  \city{}
  \country{}
}

\author{Daniele De Sensi}
\affiliation{%
  \institution{Sapienza University of Rome}
  \city{}
  \country{}
}

\author{Torsten Hoefler}
\affiliation{%
  \institution{ETH Z\"urich}
  \city{}
  \country{}
}
\affiliation{%
  \institution{Microsoft}
  \city{}
  \country{}
}

\renewcommand{\shortauthors}{Bonato et al.}

\begin{abstract}
Next-generation datacenters require highly efficient network load balancing to manage the growing scale of artificial intelligence (AI) training and general datacenter traffic. However, existing Ethernet-based solutions, such as Equal Cost Multi-Path (ECMP) and oblivious packet spraying (OPS), struggle to maintain high network utilization due to both increasing traffic demands and the expanding scale of datacenter topologies, which also exacerbate network failures. To address these limitations, we propose REPS, a lightweight decentralized per-packet adaptive load balancing algorithm designed to optimize network utilization while ensuring rapid recovery from link failures. REPS adapts to network conditions by caching good-performing paths. In case of a network failure, REPS re-routes traffic away from it in less than 100 microseconds. REPS is designed to be deployed with next-generation out-of-order transports, such as Ultra Ethernet, and uses less than 25 bytes of per-connection state regardless of the topology size. We extensively evaluate REPS in large-scale simulations and FPGA-based NICs.
\end{abstract}

\keywords{load balancing; packet spraying; out-of-order transport; RDMA; Ultra Ethernet}

\maketitle
\section{Introduction}



\tom{Modern AI training cluster networks draw on both high-performance computing (HPC) and cloud-native architectures~\cite{convergance, alibaba, facebook}. These systems rely on Remote Direct Memory Access (RDMA) for low-latency, high-throughput communication. InfiniBand, a specialized RDMA transport, is widely used for its performance~\cite{269656}. More recently, operators have adopted commodity Ethernet fabrics, especially RoCEv2~\cite{Roce}, which enable RDMA over standard hardware using off-the-shelf switches and NICs to reduce costs. However, as training clusters scale from around 10,000 endpoints (enough to train GPT-4 or Llama 3–scale models~\cite{achiam2023gpt,dubey2024llama}) to over 100,000 nodes, both interconnects face fundamental limits in sustaining peak bandwidth due to:
}
\begin{enumerate}
\setlength\itemsep{0mm}
\item increased traffic volume and burstiness in collective communication compared to traditional workloads~\cite{li2020pytorch, zhao2023pytorch},
\item the complexity and overhead of managing a lossless, in-order network at scale, especially under link failures and performance degradation.~\cite{facebook}.
\end{enumerate}

Thus, the community recognized the need for network stacks tailored to distributed training traffic while remaining compatible with commodity datacenter Ethernet infrastructure~\cite{facebook}. 
Such proposals include SRD by Amazon \cite{amazon}, Falcon by Google \cite{googleIntroducingFalcon}, TTPoE by Tesla~\cite{ttpoe}, and Ultra Ethernet (UE) \cite{ultra}, developed in collaboration between major tech players. Key open questions in these proposals are how to address \textbf{\textit{load balancing}} and \textbf{\textit{mitigate link failures}}.

Current-generation in-order Ethernet-based training systems (such as RoCEv2-based clusters) typically rely on ECMP \cite{ecmp} or similar mechanisms for decentralized routing and load balancing. ECMP load balancing logic applies a hashing function to the 5-tuple header of each data packet to determine the next hop to take. The benefit of this scheme is that, ignoring link failures, it is unlikely to receive out-of-order packets at the destination NIC as packets belonging to the same connection will be routed through the same network path. 

However, ECMP routing is fragile when different connections get hashed to the same link~\cite{hedera, conga, facebook}. In this scenario, flows can get hashed to the same path even when other paths are free, resulting in congestion and queue build up, which, in turn, can result in drops and go-back-N retransmissions~\cite{ecnsu}.

Moreover, recent works have shown that link failures drastically impact both training times and economic costs. A single link failure can have $\approx 20 \times$ higher cost impact in distributed training workloads than in cloud workloads \cite{alibaba, facebook}. This observation, along with the increasing scale these systems are growing at, highlights the need for a transport layer with a load balancing scheme that can adapt near-instantly, e.g., within a few round-trip times (RTTs), to network failures and, consequently, bandwidth asymmetries across the topology.

\changed{Several solutions have been proposed to overcome ECMP's limitations. MPTCP \cite{mptcp}, PLB \cite{10.1145/3544216.3544226}, FlowBender \cite{flowbender}, Flowlet Switching \cite{201562}, and Flowcell \cite{flowcell} divide flows into subflows or flowlets and then route each one individually. However, these solutions are still designed for in-order networks, making them inherently sensitive to collisions, prone to handling failures poorly, and requiring significant memory for each connection~\cite{10.5555/3307441.3307472}.
Load balancers, such as Oblivious Packet Spraying (OPS)~\cite{6567015} and Multi-Path RDMA (MPRDMA) ~\cite{10.5555/3307441.3307472}, that operate at a per-packet granularity, can mitigate ECMP-based collisions. However, both approaches lack mechanisms to load balance effectively in the presence of network failures. Additionally, MPRDMA is constrained by its limited support for receiving out-of-order (OOO) packets and its requirement for per-packet acknowledgments (ACKs).}

\changed{Our key insight is that the challenges of spraying-based load balancing can be addressed by \textit{adaptive} packet spraying, paired with a transport layer that natively supports \textit{out-of-order} packet delivery, as demonstrated in SRD~\cite{amazon}, UE~\cite{10.5555/3307441.3307472,hoefler2025ultraethernetsdesignprinciples}, Falcon~\cite{googleIntroducingFalcon} and recent research work~\cite{6567015, nvidia, facebook}. Based on this insight, we design and contribute to the Ultra Ethernet Consortium (UEC) our decentralized load-balancing scheme, \underline{\textbf{R}}ecycled \underline{\textbf{E}}ntropy \underline{\textbf{P}}acket \underline{\textbf{S}}praying aka \textit{REPS}. REPS caches "good" network paths in a circular buffer and quickly recovers (within a few RTTs) from network failures by adaptively discovering or freezing network paths. Figure~\ref{fig:reps_scheme} provides a schematic overview of REPS. The 1.0 Ultra Ethernet specification explicitly cites REPS as a reference load-balancing mechanism for Ultra Ethernet Transport (UET)~\cite{ultraethernet2025spec}.} 

\begin{figure}[htbp]
    \centering
    \includegraphics[width=\linewidth]{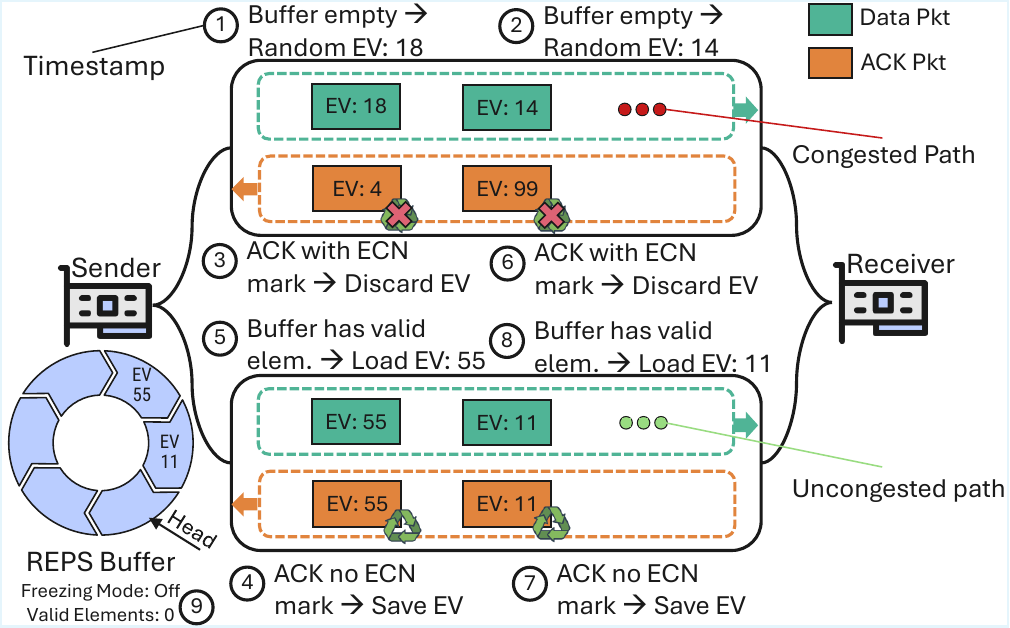}
    \caption{Schematic of REPS. We do not show elements invalidation (for EV 55 and 11) in the circular buffer after usage.}
    \label{fig:reps_scheme}
    \vspace{-1.5em}
\end{figure}

REPS does not require any special hardware support from switches beyond ECMP-based header hashing and ECN, which are standard features in modern switches \cite{ecnsu,spectrum}. REPS requires only $\approx25$ \tom{bytes of state per-connection regardless of the topology size, whereas MPTCP requires 368 extra bytes for 8 sub-flows~\cite{10.5555/3307441.3307472}.}

We extensively evaluate REPS in simulation and augment the simulation findings with real hardware results in Section~\ref{sec:eval_overall}. We deploy REPS in a cluster with modified FPGA-based RDMA-capable NICs. In large-scale network simulations, REPS consistently outperforms state-of-the-art load balancing algorithms. REPS outperforms ECMP and OPS by up to $6\times$ and $1.25\times$ in symmetric networks, and by up to $5\times$ and $2\times$ in asymmetric networks, outperforming OPS by as much as $100\times$ during short-term transient link failures.



\section{REPS Building Blocks}
\changed{In this section, we introduce the fundamental building blocks necessary to understand REPS logic. We first describe key congestion signals, such as ECN \cite{rfc3168} and packet loss. We then introduce essential terms and concepts related to load balancing, including ECMP and EV. \tom{These concepts are largely established, but we revisit them here to provide clear terminology and a basis for the discussion that follows.}}

\subsection{Congestion Signals}\label{sec:cong_signals}

\begin{description}[leftmargin=0pt]
\setlength\itemsep{0.25em}


\item[Explicit Congestion Notification (ECN) marking] allows switches to notify congestion by setting a bit in the traffic class field of the IP header. The receiver then sends back this marked ECN bit to the sender in its ACK packet header. The sender can choose to react to this congestion signal by adjusting its sending rate~\cite{zhu2015congestion,10.1145/1851275.1851192,8847013}. Switches can employ various strategies for marking packets. For instance, in \textit{Random Early Detection} (RED) \cite{251892}, switches probabilistically mark packets based on queue size, with the marking probability increasing linearly between two thresholds ($K_{min}$ and $K_{max}$). 

In REPS, we use ECN as the congestion signal to detect path congestion due to its simplicity and widespread adoption \cite{ecnsu,spectrum}. Since ECN is not marked for packets when the queue is smaller than $K_{min}$, ECN effectively filters out cases of minor queuing due to packet collisions across multiple hops, while identifying true congestion at a single bottleneck. In contrast, delay-based signals struggle to differentiate between these scenarios unless enhanced by advanced switch features, such as in-network telemetry (INT) \cite{52672}.

\item[Packet loss] has long been a key indicator of severe congestion in networks~\cite{251892,10.1145/2208917.2209336}. However, using packet losses as the sole signal for congestion detection can result in delayed responses, as losses typically indicate a point of significant congestion. Packet loss detection, often based on timeouts, can be challenging to calibrate and may lead to unnecessary retransmissions.
We categorize packet losses in two categories: losses because of congestion and losses due to networking failures. To differentiate the two types of losses and to improve reaction times, packet trimming, which only triggers for congestion drops, can be employed ~\cite{adrian2022implementing, 179729,10.1145/3098822.3098825,278348, trimming_rdma}. Trimming can be implemented on many existing switches~\cite{adrian2022implementing} and is starting to be supported by some switch vendors. It is also being pushed as a UE protocol feature \cite{ultraethernetUltraEthernet}. On current generation networks, REPS can be deployed with timeouts, and in next-generation augmented with trimming to distinguish between loss types (more details in Appendix~\ref{appendix:freezing_details}).


\item[Congestion Control (CC) algorithms] rely on congestion signals, such as ECN marks or packet loss, to adjust a flow's sending rate or window with the goal of maintaining high network utilization while preventing queue build up. REPS is designed to work well with any CC algorithm as long as they support receiving and acknowledging packets out-of-order for a given message. 
In particular, we show later how REPS works well with EQDS, a variant of DCTCP, and a proprietary internal CC algorithm \cite{278348, 10.1145/1851275.1851192}.

\end{description}
\subsection{Load Balancing}

\begin{description}[leftmargin=0pt]
\setlength\itemsep{0.25em}

\item[Equal Cost Multi-Path (ECMP)] is one of the most commonly used and simple load balancing mechanisms. It works by using a hashing function to randomly choose one of the available paths for a given packet \cite{ecmp}. The hashing function usually takes as input 5 elements (aka \textit{five tuple}) from the packet header: the protocol number, source address, destination address, source port, and destination port (some variations utilize only the four tuple without the protocol number). Recent approaches have proposed incorporating additional fields, such as the Time-to-Live (TTL) or the Flow Label (in IPv6), to further refine the hash calculation \cite{flowbender, 10.1145/3544216.3544226}.

Under normal circumstances, all packets belonging to the same flow are assigned statically to a given path since the hashing function will use the same values as input. This assignment is done statically and ignores the current network congestion and failure conditions. As a result, two or more flows might be assigned to the same path even if there are many more paths available. This will inevitably result in heavy congestion and, possibly, packet drops as a consequence. Such ECMP \textit{hash collisions} are a well-documented limitation of standard ECMP~\cite{facebook,conga,ecmp_improved}.

\item[Entropy Value (EV)] is a value in the packet header that can be configured to be an input to the ECMP hashing function in switches. This value, which is set by the sender, helps to determine a packet's path through the network. Possible header fields to be used as an EV are the Source Port field in the packet header \cite{10.5555/3307441.3307472} or the Flow Label field in IPv6 \cite{10.1145/3544216.3544226}. \tom{This is consistent with previous load balancing works and has already been demonstrated in practice \cite{flowbender, 10.1145/3544216.3544226, 201562, 10.5555/3307441.3307472}. REPS leverages entropy values (EVs) to improve load balancing and mitigate ECMP collision limitations, without requiring knowledge of the exact mapping from a packet's EV to its resulting path. Switch hash functions can produce collisions: packets with distinct EVs, with all other header fields held constant (e.g., source and destination IPs), may still traverse the same physical path. The sender, including REPS, cannot directly infer whether changing the EV will change the path. Nevertheless, with a well-designed hash function, the induced path distribution should closely approximate a uniform random assignment, yielding near-optimal spreading.
} 

\item[Entropy Values Set (EVS)] is a fixed-size set of EVs, since the number of possible values is constrained by the number of bits they can occupy in the packet header. For example, the source port field in a UDP header is assigned 16 bits, giving the EVS a size of 65536 possible values (excluding some reserved values) \cite{10.5555/3307441.3307472}. While different numbers of bits could be allocated for the EVS, we analyze in Section~\ref{sec:evs_size} how many are required for optimal performance. It is generally advantageous for an algorithm to achieve good load balancing performance with a \textit{small} EVS size since that often reduces the algorithm's memory overhead. \tom{Note that the minimum effective EVS size also depends on the network topology: smaller topologies require fewer EVs to fully explore all available physical paths.}

\item[Oblivious Packet Spraying (OPS),] also known as Random Packet Spraying (RPS)~\cite{6567015}, randomly distributes individual packets across all available paths between a sender and a receiver. This is done by selecting a random EV for every packet at the sending host or by choosing a random output port at the switches. OPS has the advantage of distributing traffic evenly across multiple paths, addressing most ECMP issues. However, it is unaware of asymmetries or failures and can still be sub-optimal even in a perfectly symmetrical network (Section~\ref{sec:eval:nonos}).

\end{description}



\section{REPS}\label{sec:reps}\changed{
\underline{\textbf{R}}ecycled \underline{\textbf{E}}ntropy \underline{\textbf{P}}acket \underline{\textbf{S}}praying aka \textit{REPS} is a load balancing algorithm that relies on simple and memory-efficient endpoint logic. The key idea behind REPS is straightforward: when congestion is detected on a certain path, we explore alternative paths while caching and reusing paths with little to no congestion. Specifically, REPS uses a circular buffer of a fixed size to cache EVs of uncongested paths. \tom{By design, REPS can be implemented in NIC hardware or firmware with minimal memory/area footprint. REPS does not need any switch support besides ECMP hashing and ECN marking. Placing REPS at the endpoint not only simplifies deployment but also provides a broader, end-to-end view of the network, unlike fabric-side approaches which are limited to local visibility.} Algorithm~\ref{alg:update_entropy} details pseudocode for a REPS sender receiving an ACK and detecting a failure. Algorithm~\ref{alg:on_send} describes pseudocode for a REPS sender sending a data packet.}

\subsection{Core Logic: Path Exploration and Reuse} \label{sec:reps_explore}

\tom{During the first Bandwidth-Delay Product (BDP) worth of packets of a new or idle connection, a REPS sender \textit{explores} random entropies from the EVS. This exploration is necessary because, initially, there is no fresh knowledge about the network's state. In this warm up phase, REPS operates similarly to OPS.} 

Upon receiving a data packet, the receiver copies the EV from the received packet into the acknowledgement (ACK) packet, forwarding it back to the sender. More specifically, ACKs can use that same EV for their own header instead of using a new header field, eliminating the need for extra header space and for any changes to the packet wire format. 

When an ACK arrives at the sender, if it is not ECN-marked, the EV it carries is cached in the circular buffer, and its validity bit is set to 1. Otherwise, if the ACK is ECN-marked, REPS does not cache the EV and discards it. When set, the validity bit indicates that an entropy has not been used after it has been added to the circular buffer.
When sending a data packet out, REPS first checks if there are any valid EVs in its buffer. In case there is any valid EV, REPS \textit{reuses} the oldest valid EV from the circular buffer and resets its validity bit. Otherwise, REPS \textit{explores} a random EV from the EVS.

The circular buffer in REPS ensures that bursts of back-to-back ACKs with "good" entropies are correctly cached and reused. Moreover, it guarantees stable load balancing in the case of failures as shown in Section~\ref{sec:reps_freezing}. We use a circular buffer of 8 elements based on empirical evidence and the bounds derived from Theorem~\ref{thm:convergence}.

\subsection{Failure Mitigation: Freezing Mode} \label{sec:reps_freezing}

\changed{When a network, running without REPS, experiences any kind of transient (e.g., link flap) or persistent unrecoverable failure (e.g., a link or switch failure), it will take the system some time to recover from it: ranging from several milliseconds to update the ECMP routing group to several seconds if a reboot is needed, and more if a swap is needed \cite{ecmp_recover1, ecmp_rec2}.} 

If we assume that it takes $10$ ms to exclude a failed cable from a routing group, packets will still be routed to this failing group during this transient period, resulting in packet drops. Specifically, with a 4 KiB MTU and a 400 Gbps link, this could potentially result in over 120,000 packets (approximately 0.5 GB) being lost (ignoring congestion control). This becomes even more critical in the case of other failures where it takes longer to update the routing.

\tom{By default, REPS \textit{Core Logic} would already limit utilizing failing paths since it would only recycle paths that are still active and returning ACKs without ECN. However, if the REPS buffer was empty, it would select a random EV that could potentially map to a failing path. To further improve the performance of REPS and prevent further losses, we introduce \emph{freezing mode}. REPS enters this mode when detecting failures via indirect feedback from the network.} REPS uses a timeout heuristic (Section~\ref{sec:cong_signals}) that can be enhanced with packet trimming (Appendix~\ref{appendix:freezing_details}) as a natural feedback from the network to detect failures along a path. 

When in \textit{freezing mode}, REPS:
\begin{enumerate}[noitemsep, topsep=0pt, partopsep=0pt]
\item avoids exploring new EVs at random since this could result in the hashing function picking a failing path,
\item reuses the elements that are currently in the circular buffer even if they might be invalid. 
\end{enumerate}

\changed{While this strategy could result in slightly worse load balancing (due to potentially reusing the same EV several times), it comes with the major benefit of guaranteeing that REPS will almost never pick the failing path again since the recent received EVs point to healthy paths.
Considering the example above: by enabling freezing mode, the number of drops decreases from over 120K packets to only about 1K.}

\changed{\tom{To decide when to exit freezing mode, we either send probe packets to check the status of the failing paths or, if probing is not available, we simply exit freezing mode after a fixed amount of time. Once we exit freezing mode, we occasionally use random EVs to allow REPS to explore new paths and assess whether we detect new packets failing or not. This prevents REPS from getting stuck in a suboptimal state if the EVs in the buffer were all pointing to a failing path, a rare scenario that can theoretically happen with properly timed back-to-back network failures. Moreover, if REPS exits freezing mode before the issue is fully resolved, it will simply re-enter the mode shortly afterward with minimal impact on performance and a very small impact on extra packet drops (as we will show in Section~\ref{sec:eval:fail}).}}

The intuition behind freezing mode is that once we suspect there is a failure, we want to start avoiding it as soon as possible. Interestingly, we observe that even if we enter freezing mode unnecessarily (i.e., by mistaking a congestion drop for a network failure), REPS would still load balance well, as discussed in Section~\ref{sec:evs_size} and Appendix~\ref{appendix:freezing_details}. This observation means that even if there is doubt about whether a real failure occurred, REPS can be conservative and can safely enter freezing mode. Finally, we note that freezing mode is not crucial to the inner working of REPS, as it can work also without it. It merely adds more guarantees on top of it and improves performance in some failure cases. We showcase the impact of running REPS without freezing mode in Appendix~\ref{appendix:additional_results:freezing}.

\begin{algorithm}[ht]
\footnotesize
    \captionsetup{font=small}
    \caption{\small REPS logic upon ACK receive and failure detection.}
    \label{alg:update_entropy}
    \begin{algorithmic}[1]
        \State $\mathit{repsBuffer} = []$ \Comment{State variables.}
        \State $\mathit{isFreezingMode} = false$  
        \State $\mathit{head, numberOfValidEVs, exploreCounter} = 0$  
        \State

        \Procedure{onAck}{ackPacket} \Comment{\textit{Sec.~\ref{sec:reps_explore}}}
            \If{\textit{ackPacket.ecn is set}}
                \State \Return
            \EndIf
            \If{\textit{not repsBuffer[head].isValid}}
                \State $\mathit{numberOfValidEVs}++$
            \EndIf
            \State $\mathit{repsBuffer[head].cachedEV} = \mathit{ackPacket.ev}$
            \State $\mathit{repsBuffer[head].isValid} = \textit{true}$
            \State $\mathit{head} = (\mathit{head} + 1) \% \mathit{REPS\_BUFFER\_SIZE}$
            \If{\textit{isFreezingMode} and \textit{now()} > \textit{exitFreezingMode}} \Comment{\textit{Sec.~\ref{sec:reps_freezing}}}
                \State $\mathit{isFreezingMode} = \textit{false}$
                \State $\mathit{exploreCounter} = \textit{NUM\_PKTS\_CWND}$
            \EndIf
        \EndProcedure

        \State 

        \Procedure{onFailureDetection()}{} \Comment{\textit{Sec.~\ref{sec:reps_freezing}}}
            \If{\textit{not isFreezingMode} and \textit{exploreCounter} == 0}
                \State $\mathit{isFreezingMode} = \textit{true}$ 
                \State $\mathit{exitFreezingMode} = \textit{now()} + \textit{FREEZING\_TIMEOUT}$ 
            \EndIf
        \EndProcedure
    \end{algorithmic}
\end{algorithm}

\begin{algorithm}[ht]
\footnotesize
    \captionsetup{font=small}
    \caption{\small REPS logic on send datapath.}
    \label{alg:on_send}
    \begin{algorithmic}[1]
        \State
        \Comment{Variables already listed in Algorithm~\ref{alg:update_entropy}}
        \Procedure{getNextEV()}{}
            \If{$\mathit{numberOfValidEVs} > 0$}
                \State $\mathit{offset} = (\mathit{head} - \mathit{numberOfValidEVs}) \% \mathit{REPS\_BUFFER\_SIZE}$
                \State $\mathit{repsBuffer[offset].isValid} = \textit{false}$
                \State $\mathit{numberOfValidEVs} = \mathit{numberOfValidEVs} - 1$
            \Else \Comment{Must be in freezing mode.}
                \State $\mathit{offset} = \mathit{head}$
                \State $\mathit{head} = (\mathit{head} + 1) \% \mathit{REPS\_BUFFER\_SIZE}$
            \EndIf
            \State \Return $\mathit{repsBuffer[offset].cachedEV}$
        \EndProcedure

        \State

        \Procedure{onSend}{dataPacket}
            \If{$\mathit{exploreCounter} > 0$}
                \If{$(\mathit{--exploreCounter} \% \mathit{REPS\_BUFFER\_SIZE}) == 0$}
                    \State $dataPacket.ev = \textit{rand()} \bmod \mathit{EVS\_SIZE}$
                \EndIf
            \ElsIf{\textit{repsBuffer.isEmpty()} \textbf{ or } \big(\textit{numberOfValidEVs} == 0 \textbf{ and } \textit{not isFreezingMode}\big)}
                \State $\mathit{dataPacket.ev}$ = \textit{rand()} \% $\mathit{EVS\_SIZE}$
            \Else
                \State $\mathit{dataPacket.ev}$ = \textit{getNextEV()}
            \EndIf
        \EndProcedure
    \end{algorithmic}
    \vspace{-0.4em}
\end{algorithm}

%


\subsection{REPS Design Advantages}
\begin{description}[leftmargin=0pt]
\setlength\itemsep{1mm}

\item[Simple and versatile algorithm:] \changed{The simplicity of REPS enables cost-efficient hardware support, as it does not require any change of the packet headers format or existing network components. Moreover, its code is short and easy to implement and understand. REPS works best with per-packet ACKs, but we show in Section~\ref{sec:ack_ratio} that it still performs well even with ACK coalescing.} While we focus on fat tree topologies in this paper, we believe REPS could work well even with different topologies \cite{dragonfly,10.5555/3571885.3571899} with little to no adjustments required. Moreover, we expect REPS to work well even when using source-based routing algorithms where the NIC chooses directly the paths to use instead of relying on ECMP-hashing in switches. In this case, REPS would store the actual path ID for its EVs.
 
\item[Minimal NIC memory footprint:]  A key advantage of REPS is that it does not need to track per-EV metrics and statistics. As will be discussed in Section~\ref{sec:evs_size}, achieving good performance with OPS requires a relatively large EVS. If OPS were to maintain metrics for each EV, the memory overhead would be excessive for a hardware NIC implementation, e.g., 64 Kib to store 1 bit per entropy value for an EVS with 64K EVs. However, REPS only needs a fixed number of bytes in memory. More specifically, as detailed in Table~\ref{tab:memory_footprint}, REPS requires only around $\approx25$ bytes. Moreover, even when constrained to a small EVS, REPS is still able to perform well (Section~\ref{sec:evs_size}), which can further reduce REPS' memory footprint by 1 byte since Table~\ref{tab:memory_footprint} assumed 16 bits per EV. 

There is a subtle observation as to why REPS can achieve a great performance without needing a lot of state: while the REPS buffer is used to cache good entropies, it is really only useful in certain scenarios like when receiving a burst of ACKs or during freezing mode. In reality, most of REPS' state is on the wire, stored in the inflight data and ACK packets, which will inform REPS about the good paths in the network.

\begin{table}[ht]
\centering
\scriptsize 
\begin{tabular}{p{0.6\linewidth} p{0.2\linewidth}}
    \toprule
    \textbf{Component} & \textbf{Footprint (bits)} \\
    \midrule
    \multicolumn{2}{l}{\textbf{Circular Buffer Element (× elements in buffer):}} \\
    \quad Entropy Value \ (\textit{cachedEV}) & 16 \\
    \quad Entropy Validity Bit \ (\textit{isValid}) & 1 \\
    \midrule
    \multicolumn{2}{l}{\textbf{Global Variables:}} \\
    \quad Head Buffer \ (\textit{head}) & 8 \\
    \quad Number Valid Entropies \ (\textit{numberOfValidEVs}) & 8 \\
    \quad Exit Freezing Time \ (\textit{exitFreezingMode}) & 32 \\
    \quad Is Freezing Mode \ (\textit{isFreezingMode}) & 1 \\
    \quad Explore Counter \ (\textit{exploreCounter}) & 8 \\
    \midrule
    \textbf{Total (1 element in buffer)} & \textbf{74 $\approx$10 bytes} \\
    \textbf{Total (8 elements in buffer)} & \textbf{193 $\approx$25 bytes} \\
    \bottomrule
\end{tabular}
\caption{Per-connection memory footprint of REPS.}
\label{tab:memory_footprint}
\end{table}

\item[Quick failure mitigation:] The general approach of REPS is that it only keeps track of good paths and avoids keeping statistics on congested or failing paths. This approach enables it to promptly load balance away from a congested link or failing link as, especially for the latter, it is never going to take a random guess once a link is failing. Any alternative method that tries to avoid selecting a failing path by tracking bad EVs would need to keep records of all the EVs that map to that path for a given flow, which would involve tracking not only the failing EVs but also all those still in flight, resulting in excessive NIC's memory usage.

\end{description}


\section{Evaluation}\label{sec:eval_overall}
\tom{Our evaluation includes both large-scale simulations designed to stress test REPS under a variety of workloads, and experiments on real hardware at meaningful scale. The goal is to answer the following research questions:
\begin{itemize}
\setlength\itemsep{0mm}
    \item Does REPS outperform state-of-the-art load balancers in a healthy network with a symmetric topology? (Sections~\ref{sec:eval:nonos} and \ref{sec:hw_eval:symm})
    \item How does REPS perform in the presence of topology asymmetries or when coexisting with non-REPS traffic? (Sections~\ref{sec:eval:reps} and \ref{sec:hw_eval:asymm})
    \item Can REPS recover quickly from failures? Is it resilient to severe or widespread failures? (Sections~\ref{sec:eval:fail} and \ref{sec:hw_eval:fail})
    \item Does REPS maintain performance under varying network parameters and configurations? (Section~\ref{sec:reps_flex})
\end{itemize}}

\subsection{Evaluation Setup}
\begin{description}[leftmargin=0pt]
\setlength\itemsep{0mm}

\item[Baseline load balancers:] In large-scale simulations, we compare REPS with ECMP~\cite{ecmp}, OPS~\cite{6567015}, PLB~\cite{10.1145/3544216.3544226}, MPRDMA \cite{10.5555/3307441.3307472}, Flowlet Switching~\cite{201562}, an MPTCP-like algorithm~\cite{mptcp}, a bitmap approach where we keep per EV statistics similarly to STrack~\cite{le2024strackreliablemultipathtransport}, and adaptive RoCE by NVIDIA~\cite{nvidia}. \tom{We note that this list includes both sender-based approaches and in-network switch-based approaches}. We configure PLB to have more aggressive parameters similarly to FlowBender to improve its performance \cite{flowbender}. For Flowlet Switching we set an aggressive flowlet timeout at half of the RTT. For MPTCP, we divide each message into 8 subflows and route each one individually similarly to using multiple QPs (Queue Pairs)~\cite{facebook}.

\item[NIC congestion control:] In all simulated baseline runs, we use the same DCTCP \cite{10.1145/1851275.1851192} variant used in MPRDMA \cite{10.5555/3307441.3307472}. It applies per-ACK congestion window updates, allows the receiver to accept and acknowledge out-of-order packets, and reduces the congestion window by one MTU in case of packet drops. In the FPGA-based experimentation, we use a similar but proprietary CC algorithm that relies on ECN marking, congestion notification packets, and per-flow congestion window adjustments.


\item[Network setup:]
\tom{Regardless of the workloads that we discuss in Section~\ref{sec:workloads}, in the evaluation we simulate 3 different scenarios: (1) healthy symmetric network conditions, (2) asymmetric network conditions, which can be due to asymmetric path bandwidth or due to ECMP-hashed background traffic that increase load on specific paths (a scenario that can happen in incremental deployments of REPS), and (3) a network encountering various failure modes. We focus on the most relevant failure scenarios for real-world deployments~\cite{facebook} and we report the remaining ones in Appendix~\ref{appendix:additional_results}.}

\item[Simulation model:] \changed{We implement REPS by extending the \textit{htsim} packet-level network simulator~\cite{10.1145/3098822.3098825}. Our simulations consider different fat-tree topologies with 1024 nodes and 128 nodes and with different levels of oversubscription ranging from 1:1 (no oversubscription) to 4:1. We test 2- and 3-tier fat trees (TOR or Top-of-rack as \textit{T0}, Aggregate as \textit{T1} and Core \textit{T2}). Such topologies are commonly deployed in production datacenters designed for distributed training~\cite{43837, facebook}. We use the larger topologies for the microbenchmark runs while the smaller topology for running the datacenter traces and AI collectives.}

We reflect the specifications of current-generation switches in simulation parameters: a 4 KiB MTU size, a bandwidth of 400 Gbps, and a switch traversal latency of 500 ns~\cite{toma5,9355230}. We assume uniform link lengths and latencies, with each link exhibiting a latency of 500 ns. We set the retransmission timeout (RTO) to $70\;\mu\text{s}$ which is the amount of time it takes to traverse every queue in the network if it was full plus the network-wide RTT.
For each queue $K_{min}$ is set to 20\% of the queue size (one BDP) and $K_{max}$ to 80\% of it.

\item[REPS-FPGA:] \tom{To demonstrate the effectiveness and resilience of entropy recycling in a real network environment, we also evaluate REPS in an end-to-end setting. We use a modified production-grade FPGA-based RDMA-capable NIC running a custom transport protocol capable of handling OOO packets by directly placing the payload in memory via RDMA. The transport implementation relies on SACKs that carry bitmaps (256-bit-wide in these experiments) relaying packet delivery updates to track completions and handle selective retransmissions. Each connection maintains an 8-entry-deep REPS buffer that stores UDP source ports (2B). Our experimental setup is configured to support up to 256 connections for a total memory footprint of 4KB. Given that only a single connection is active at a time, all REPS buffers can be hosted in a single SRAM memory and the associated logic to manage REPS buffers can be multiplexed across all connections, occupying negligible resources (<0.04\% of the total device logic resources).}

The testbed consists of a two-tier fat-tree Ethernet network with 100G NICs and 12.8T switches. The default MTU for our FPGA NICs is 8KB and typical RTT incorporating NIC buffer delay and ACK processing through T0 and T1 are in the order of 10 and 15 us, respectively.

\end{description}

\subsection{Workloads}\label{sec:workloads}
We evaluate REPS on a mix of synthetic benchmarks, real datacenter traces, and distributed training collectives. 
\begin{description}[leftmargin=0pt]
\setlength\itemsep{0mm}
\item[Synthetic benchmarks set] consists of (1) \textit{incast}, (2) \textit{permutation}, and (3) \textit{tornado} traffic patterns. Incast happens when multiple senders simultaneously send to one receiver. It is very common in storage workloads \cite{storage1, storage2} but also, with a small incast degree, in distributed training \cite{facebook}. In the permutation pattern, each node sends to a random receiver, and we ensure that each node is sending and receiving to exactly one node \cite{10.1145/1851275.1851192}. The tornado pattern is a special case of the permutation where each node sends to its "twin" node in the other half of the tree. For example, with 128 nodes, node 0 would send to 64 and vice-versa, node 1 to 65 and so on. Tornado is an important worst case for load balancing, as each packet is required to traverse the full tree \cite{tornado}.

\item[Datacenter traces:] We use real datacenter traces from similar previous work \cite{10.1145/1851275.1851192, hermes}. We use a series of traces used for web search in production clusters. In such distribution the majority of flows are quite small (less than 100 KB) while a small number of flows are large. For each node we select randomly the receiver. \changed{More details are available in Appendix~\ref{appendix:additional_data}.}

\item[AI collectives:] We show simulated results for two commonly used collectives in AI training: the Allreduce implemented via the ring and butterfly algorithm~\cite{li2020pytorch}, the Alltoall implemented using an algorithm where we limit the number of parallel connections per node ($n$ connections) \cite{10.5555/3571885.3571899, naumov2019deep}.

Our baseline traffic for REPS-FPGA consists of 128 endpoints under two T0 switches continuously performing 4 MB ring-based AllReduce collective operations, with the logical ring laid out such that all connections traverse the T1 spine to maximize the pressure on the spine of the topology. 

\end{description}

\subsection{Simulation Results}~\label{sec:simulation_results}
We conduct a detailed analysis of REPS behavior for each network condition (see Section~\ref{sec:eval_overall}). We first examine a specific case in depth and then summarize key takeaways.
\subsubsection{Healthy Symmetric Network Conditions}\label{sec:eval:nonos}
In this Section we evaluate the performance of REPS in a simple setting where there is no oversubscription and there are no failures, meaning the network is perfectly symmetrical. Intuitively, this seems the best situation for oblivious packet spraying since evenly splitting the packets across multiple links should result in the best performance. \changed{However, as we will see later in this section and based on our theoretical model in Section~\ref{sec:theo_mot}, this is not the case as REPS still offers an up to 25\% advantage over OPS. This is because of ECMP collisions that still happen with OPS. While over a long period of time, each link will be evenly used, there will still be short-term collisions happening that will increase and decrease the link utilization of certain links.} 

\begin{description}[leftmargin=0pt]
\setlength\itemsep{0.25em}

\begin{figure}[htbp]
    \centering
    \includegraphics[width=\linewidth]{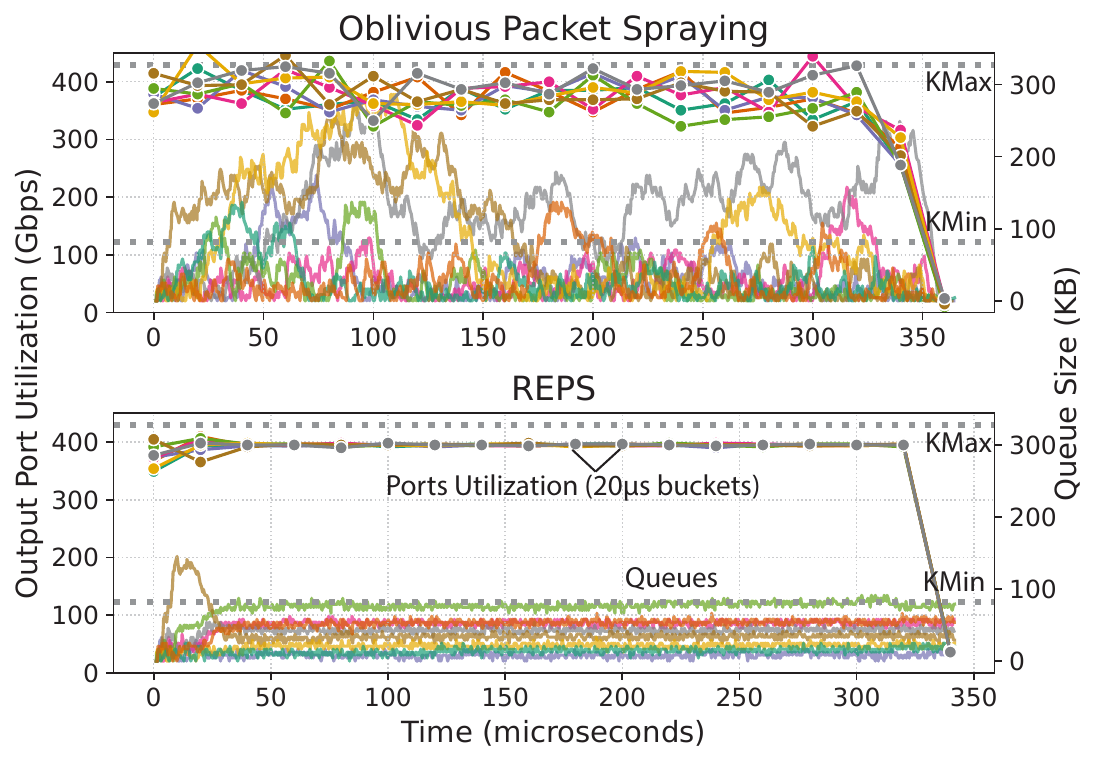}
    \caption{A tornado workload with 16 MiB messages using OPS (top) and REPS (bottom) as load balancers.}
    \label{fig:nocc_symm}
    \vspace{-1em}
\end{figure}

\begin{figure*}[h]
    \centering
    \includegraphics[width=\textwidth]{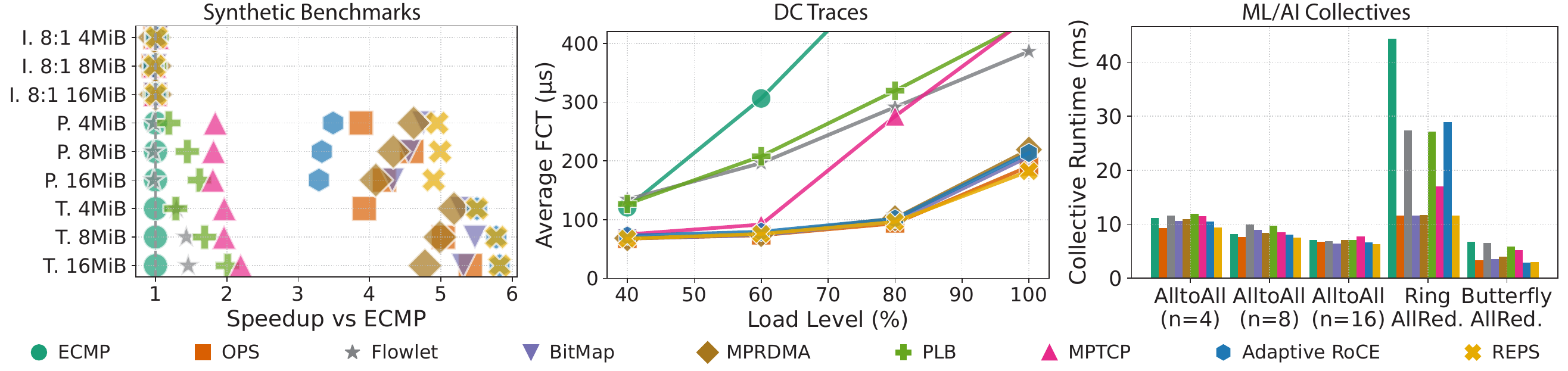}
    \caption{REPS performance in synthetic benchmarks (I.=Incast, P.=Permutation, T.=Tornado), DC traces and AI collectives.}
    \label{fig:macro_symm}
    \vspace{-1em}
\end{figure*}

\begin{figure}[t]
    \centering
    \includegraphics[width=\linewidth]{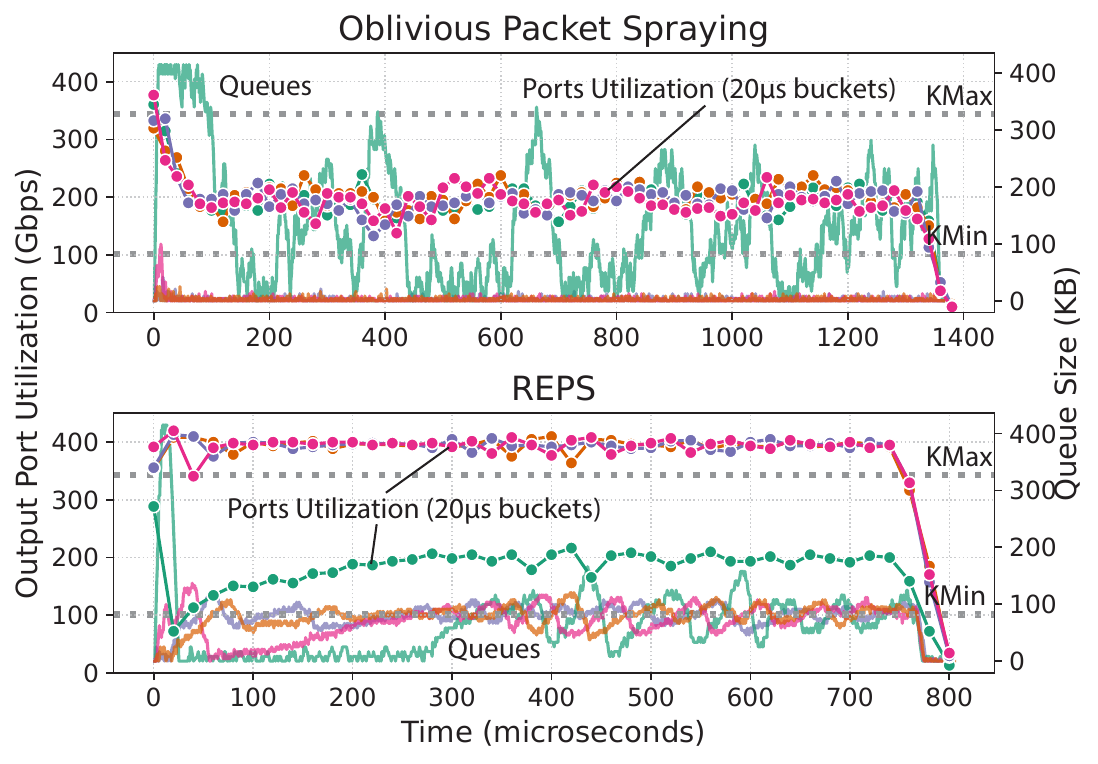}
    \caption{REPS vs. OPS in a 32 MiB message send with an asymmetric topology.}
    \label{fig:micro_symm}
    \vspace{-1em}
\end{figure}

\item[Microscopic analysis:]
We study this effect in a tornado pattern, which, theoretically, can be managed entirely by optimal load balancing. We study what happens at a TOR switch and register the link utilization of the uplinks both over the entire simulation and also at smaller time buckets. For visualization purposes we limit these runs to a 2-tier network where each switch has 8 uplinks. However, we note that this problem is present, to an even bigger degree, when using a larger number of uplinks as explained in Section~\ref{sec:theo_mot}.

In Figure~\ref{fig:nocc_symm} we visualize statistics for a single T0 switch during the workload run. In particular we show two key metrics over time: 1) on the left Y-axis we show the \textit{output port utilization} at fixed time intervals of 20 us. If it goes above 400 Gbps it means that during the studied time bucket some queueing was created. If it goes below it means the output port was slightly under-utilized. 2) On the right Y-axis we show the \textit{queue size} over time of the 8 output ports.

In the case of OPS, we can see that due to the random nature of it, queues are created over time and that the link utilization of each port, at small timeframes, sometimes goes significantly above and below the optimal point (15\% more or less). This shows that while OPS does still a decent job at completing the workload close to the ideal completion time, it does inevitably create unpredictable queues (potentially even exceeding $K_{max}$ and causing drops) over short time periods. Such queues will cause the CC to kick in and slightly reduce the sending rate and average output port utilization.

In the same configuration with REPS, we notice a major difference in the bottom plot of Figure~\ref{fig:nocc_symm}. In particular REPS converges quickly to a configuration where each queue is kept below $K_{min}$ (note that the only guarantee here is that all queues will be below $K_{min}$, not necessarily all at the same value). At the same time, we can also notice that all the ports converge to the perfect selection rate of 400 Gbps. While the overall completion time is only about 4\% better than OPS, the smaller queues provide a better guarantees for system low latency traffic. Moreover, as we can see in Figure~\ref{fig:macro_symm}, this gap expands as we increase the message size.

Looking at the port selection rate over the entire run of the workload for OPS vs. REPS, we observe that they are nearly equivalent. This is, again, because the problem is with the short-term collisions that OPS can experience at microscopic scale.

\item[Macroscopic analysis:]
We now focus more on the overall view comparing REPS with all the other state of the art algorithms in a series of benchmarks. 

\begin{figure*}[t!]
    \centering
    \includegraphics[width=\textwidth]{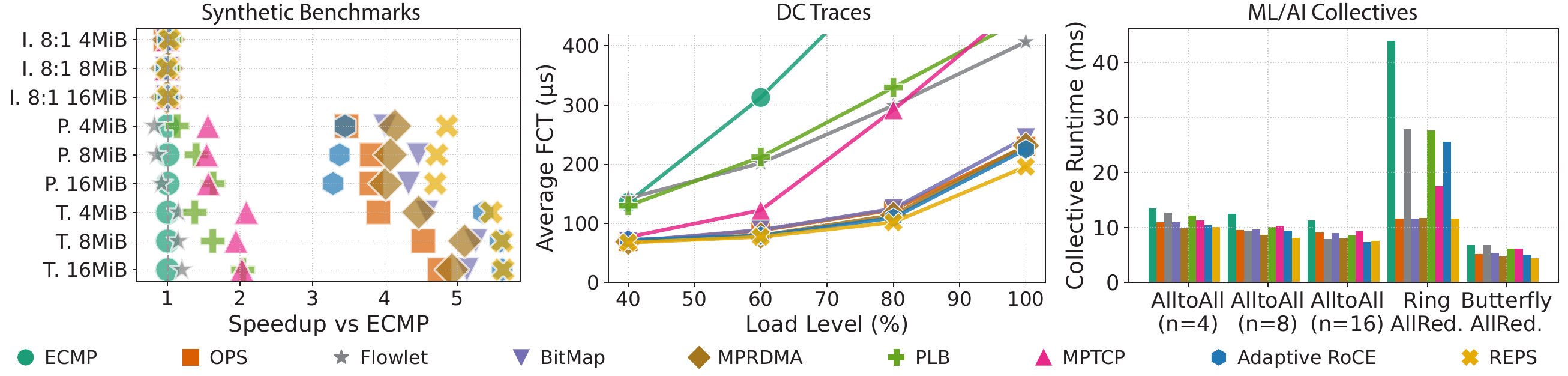}
    \caption{REPS performance in synthetic benchmarks (I.=Incast, P.=Permutation, T.=Tornado), DC traces and AI collectives and an asymmetric network due to 3\% of the TOR uplinks having downgraded bandwidth.}
    \label{fig:asymm_overall}
    \vspace{-1em}
\end{figure*}

In Figure~\ref{fig:macro_symm}, we visualize a summary of the performance of the various algorithms by looking at the runtime of the workloads (max FCT). As expected, in the case of incast, the performance is driven almost exclusively by the CC and, hence, we do not see any major difference between all the load balancers and even ECMP performs well. However, once we move to permutation and tornado workloads, ECMP collisions start to drastically reduce the performance of ECMP. In most cases REPS outperforms all the other algorithms. 

In the tornado case, Adaptive RoCE is able to match REPS since this is the ideal scenario for it: REPS, unlike Adaptive RoCE, still needs to guess during its initial BDP worth of packets. On the other hand, REPS outperforms Adaptive RoCE in the permutation pattern where taking a local optimal decision might not always lead to the best global outcome. We also see the difference between algorithms that were designed to reduce the number of out-of-order packets versus algorithms that do not have such hard constraints. Additionally, there is a distinction between algorithms that operate at packet-level granularity, such as REPS, OPS, BitMap, and MPRDMA, and those that operate at a coarser granularity, such as Flowlet and PLB.

For datacenter traces, we analyze the results for different load levels, ranging from 40\% to 100\%. Here we observe again, the clear difference between per-packet algorithms and less granular options. Even at higher load REPS is able to work well with a 5\% advantage over OPS.

We also run distributed collectives and report their completion times. We observe how, by design, the ring AllReduce has the same performance for most of the load balancing algorithms utilized. This is because, due to its ring design, there is no opportunity for congestion to accumulate. In AllToAll, REPS gets an up to 20\% advantage over the alternatives.

\end{description}

\begin{figure}[htbp]
    \centering
    \includegraphics[width=\linewidth]{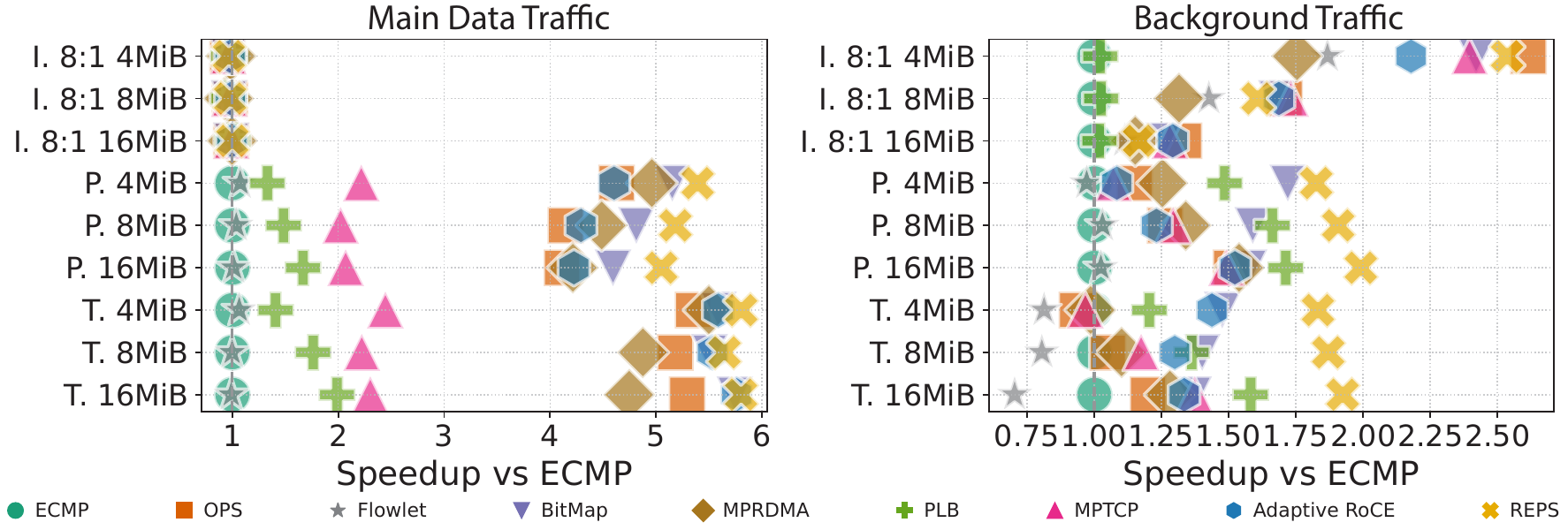}
    \caption{\tom{Synthetic traffic with REPS and non-REPS (ECMP) traffic, emulating a mixed traffic scenario.}}
    \label{fig:background_traffic}
    \vspace{-1em}
\end{figure}

\subsubsection{Asymmetric Network 
Conditions}\label{sec:eval:reps}

We evaluate REPS under different scenarios where some degree of asymmetry is created. We focus on two scenarios: 1) the network is not perfectly symmetrical because of some missing (or degraded) cables, \tom{2) there is some background traffic in the network that is using ECMP routing, analyzing the coexistence between REPS and non-REPS traffic.}

\begin{description}[leftmargin=0pt]
\setlength\itemsep{0mm}

\item[Microscopic analysis:]
We visualize this problem with a simple scenario where we have a switch with $n$ input and output ports and $n$ flows active, each from a different source sending a 32 MiB message. To create an asymmetry, we reduce one of the uplinks speed to 200 Gbps while all the other links remain at 400 Gbps. In Figure~\ref{fig:micro_symm} we visualize the output port utilization rate for OPS and REPS. 

We observe that while OPS chooses each port equally, irrespective of its actual bandwidth, REPS eventually converges to a stable configuration where the slower uplink is used less frequently. This results in both more stable queues but, more importantly, a much faster completion time ($1400\;\mu\text{s}$ for OPS and $799\;\mu\text{s}$ for REPS). 

\item[Macroscopic analysis:]
We now shift our focus to more general results when encountering asymmetries in a network. For space constraints, we focus mostly on the case where some of the links have a lower sending rate. In our first experiment we run synthetic benchmarks where $3\%$ of the TOR uplinks, chosen randomly, have been downgraded to 200 Gbps. In Figure~\ref{fig:asymm_overall}, we can see results similar to before where REPS gets an up to 500\% advantage over ECMP and 10\% advantage over the second best algorithm (usually BitMap). In the DC traces we can see a higher difference due to the asymmetry in the network. At 100\% load REPS gets a 25\% advantage over the second best algorithm and a 1000\% advantage over ECMP. We show the results for several AI collectives. We note how for AllToAll REPS keeps a small but significant advantage and in the AllReduce a sizable 30\% advantage over the second best performing algorithm.

In Figure~\ref{fig:background_traffic}, we showcase one example of REPS sharing traffic together with background ECMP traffic (we assume 10\% of the traffic is ECMP). In this case, REPS: 1) shifts REPS traffic away from ECMP traffic in order to not slow down REPS traffic, 2) helps background traffic by ensuring that it will not be slowed down by REPS traffic. This also highlights the possibility of incrementally deploying REPS on ECMP-based systems. Finally, we note that WCMP \cite{wcmp} could be used to enhance ECMP performance in the case of a topology with known asymmetries, but would not help as much in the case of unpredictable mixed traffic or sudden temporary asymmetries.

\end{description}
\subsubsection{Network Failures}\label{sec:eval:fail}
We focus our attention to cases where the network encounters a failure during operation. We collect data from several previous works on networking failures and also study internal logs to simulate the most commonly reported cases \cite{fail_source1, fail_source2, fail_source3}. Since the probability of a failure happening during a short simulation is low, we simulate worst case scenarios where we force individual failures to happen during the runtime of our simulations. In particular, we focus on total or partial failures of cables and switches. We note that we limit our failures to components that would not prevent the workload from completing (single point of failure).
\begin{figure}[ht]
    \centering
    \includegraphics[width=\linewidth]{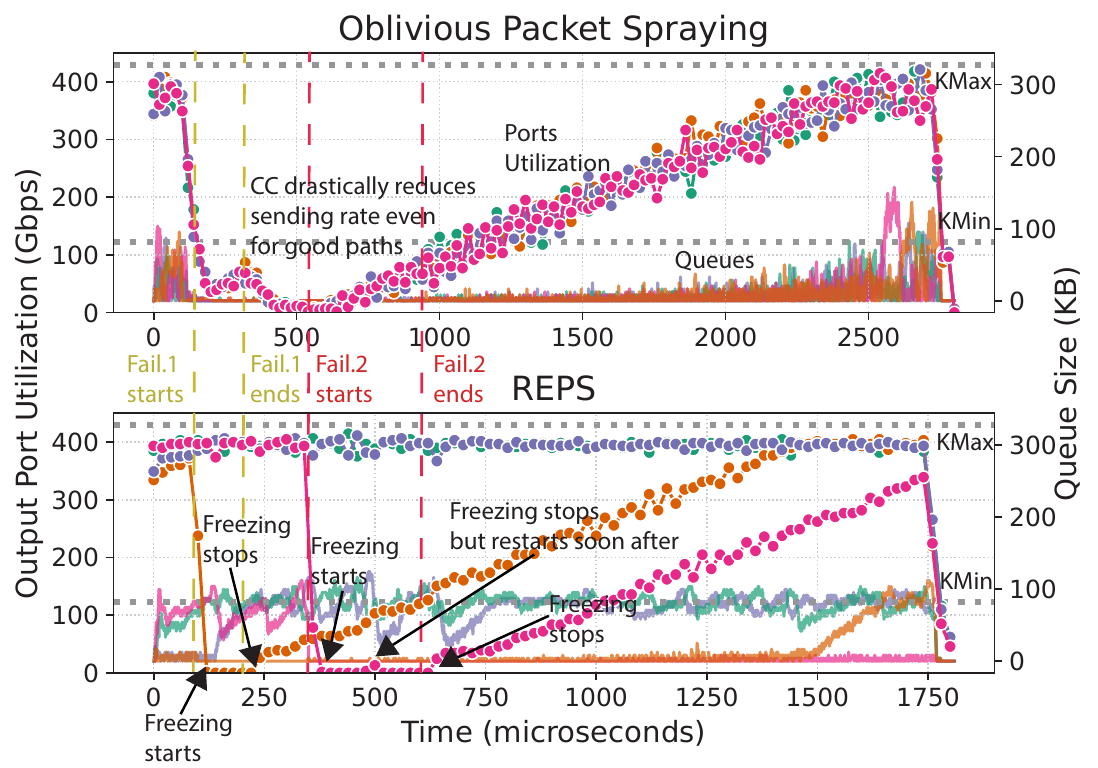}
    \caption{\changed{REPS vs. OPS in a 64 MiB permutation with two cable failures (a shorter one and a longer one).}}
    \label{fig:micro_fail}
    \vspace{-1em}
\end{figure}
\begin{description}[leftmargin=0pt]
\setlength\itemsep{0mm}
\item[Microscopic analysis:]
\tom{We begin with a synthetic scenario in which two TOR-switch uplink failures occur at different times: one lasting $100\,\mu\mathrm{s}$ starting at $t = 100\,\mu\mathrm{s}$, and another lasting $200\,\mu\mathrm{s}$ starting at $t = 350\,\mu\mathrm{s}$.}

In Figure~\ref{fig:micro_fail}, OPS keeps choosing all paths equally (although at lower rate due to CC activation), while REPS, once it enters freezing mode, stops selecting the failing paths all together after only one timeout period (order of tens of microseconds). Afterwards, once the failure stops, REPS also exits freezing mode and converges once again to use all paths. The overall result is that, compared to OPS, REPS completes the workload more than 35\% faster even with such a short failure and, more importantly, reduces the number of dropped packets by $2.5\times$. We showcase a similar analysis for incremental and concurrent failures in Appendix~\ref{appendix:additional_results:incr_fail}. 

\begin{figure*}[t!]
    \centering
    \includegraphics[width=\textwidth,height=0.7\textheight,keepaspectratio]{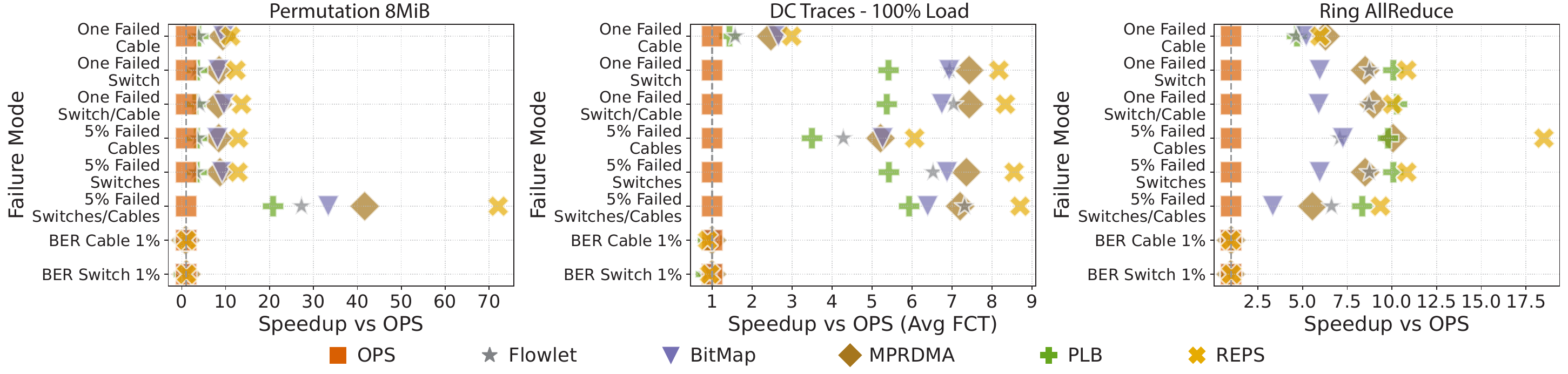}
    \caption{REPS performance under different failure modes in a 8 MiB permutation, DC traces at 100\% load and a ring AllReduce.}
    \label{fig:macro_fail}
    \vspace{-1em}
\end{figure*}

\item[Macroscopic analysis:]We showcase three cases with a series of failure modes in Figure~\ref{fig:macro_fail}. \tom{We introduce each failure mode after a fixed interval to highlight how REPS dynamically reacts to changing network conditions, and then let the failure persist for the remainder of the experiment.} We can see when dealing with total failures that REPS provides a dramatic speedup over OPS but also other load balancer algorithms. \tom{This is because of \emph{freezing mode} which helps REPS stick to the safe path configuration in its cache, allowing it to react immediately upon failure, much faster than the time ECMP routing takes to converge and exclude the failing path.} Positively, we note that the gains with REPS increase with the amount of failures. Furthermore, random drops do not affect negatively REPS performance. MPRDMA also performs well on average thanks to its self-clocking mechanism.

\begin{figure}[htbp]
    \centering
    \includegraphics[width=0.865\linewidth,height=0.22\textheight,keepaspectratio]{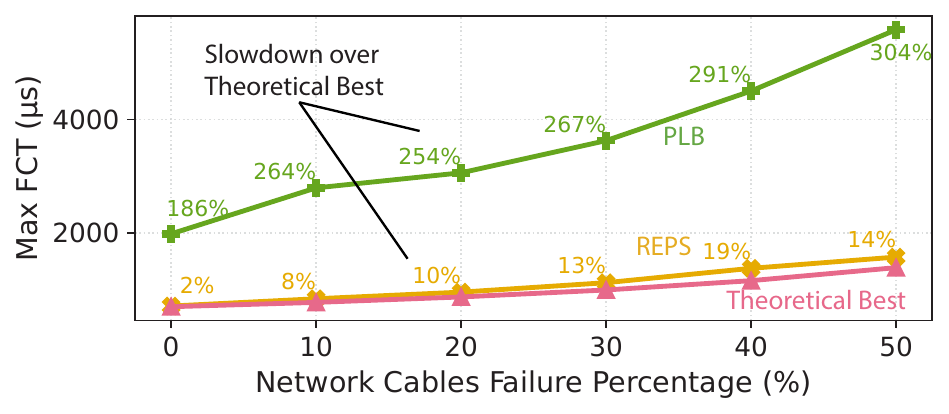}
    \caption{\changed{Extreme failures scenario.}}
    \label{fig:big_fail_case}
    \vspace{-1em}
\end{figure}

To further demonstrate the resilience of REPS, we evaluate its performance under an extreme and unlikely scenario characterized by increasingly large and long-lasting network failures during a permutation. As shown in Figure~\ref{fig:big_fail_case}, REPS performs close to an ideal load balancer, even with 50\% of network cables failing, while PLB, the second-best alternative, lags significantly behind.

\tom{This analysis proves that REPS is able to guarantee good load balancing performance under short-lived failures, long-lived failures and a combination of the two even with catastrophically high failure rates.}
\end{description}

\subsection{REPS-FPGA Evaluation} \label{sec:testbed}





\subsubsection{Healthy Symmetric Network Conditions}\label{sec:hw_eval:symm}

The first set of results focuses on a baseline healthy symmetric network configuration. Figure~\ref{fig:reps_symmetric_real:a} shows per-flow goodput defined as end-to-end useful bit rate observed by application, after header, overheads, retransmissions, etc. We use OPS and REPS across two experimental configurations: denoted as \textit{setup-1} and \textit{setup-2}. In \textit{setup-1}, all FPGA endpoints under the two T0s are active while in \textit{setup-2}, 40 out of 64 FPGA endpoints are active.

We present results from both of these configurations as we observe small unexpected performance variations depending on which and how many switch ports are active in experiments near the peak network performance levels using all switch ports (\textit{setup-1}). These variations appear to be related to internal switch microarchitectural details such as port-buffer affinity and vendor-specific scheduling policies. The \textit{setup-2} uses a subset of the switch ports and eliminates most of these vendor-specific and implementation-related behaviors. To get a better understanding of this behavior we performed a sweep where we capped the TX rate of our FPGA NICs and discovered that the slight degradation for \textit{setup-1} when using REPS appears to subside if the TX rate is capped at 95 Gbps.


\subsubsection{Asymmetric Network Conditions}\label{sec:hw_eval:asymm}
We evaluate the performance of REPS under asymmetric network conditions. We connect 16 endpoints through two T0 switches (8 endpoints each) with a total of 4 links to a pair of T1 switches. To demonstrate the adaptive load balancing capabilities of REPS, we change the link speed of one T0-T1 link from 400 Gbps to 200 Gbps, creating asymmetry in the network. Figure~\ref{fig:reps_symmetric_real:b} shows the per-flow goodput as observed by the application while Figure~\ref{fig:reps_asymmetric_real:a} the FCT distribution. OPS sends packets across all paths (including those crossing the 200 Gbps link) with equal probability and is ultimately capped by the slower 200 Gbps path. The ECN marking on the 200 Gbps path causes the CC algorithm to throttle all flows and eventually match the capacity of that single slower link, thus leading to underutilization of the remaining 400 Gbps links (that are running at 50\% utilization).

REPS can gracefully adapt in such a scenario as the cached entropies will reflect the network asymmetry and result in a path distribution that is skewed to tailor to the relative capacity of the available paths. In this example, REPS can reach high utilization with average per-flow goodput within 5\% of the ideal fair-share target. 


\begin{figure}[t]
    \begin{subfigure}[b]{0.48\linewidth}
        \includegraphics[width=\linewidth]{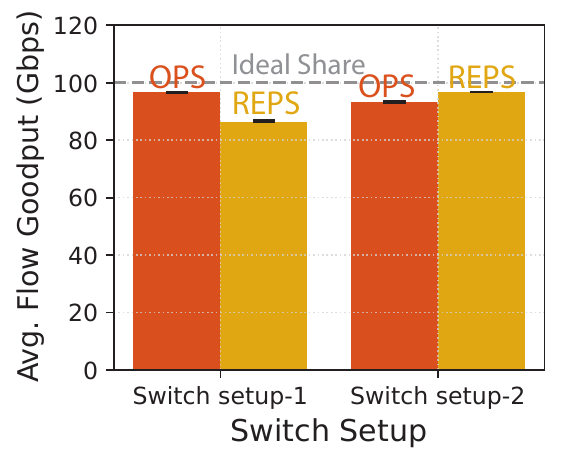}
        \caption{Symmetric Network}
        \label{fig:reps_symmetric_real:a}
    \end{subfigure}
    \hfill
    \begin{subfigure}[b]{0.48\linewidth}
        \includegraphics[width=\columnwidth]{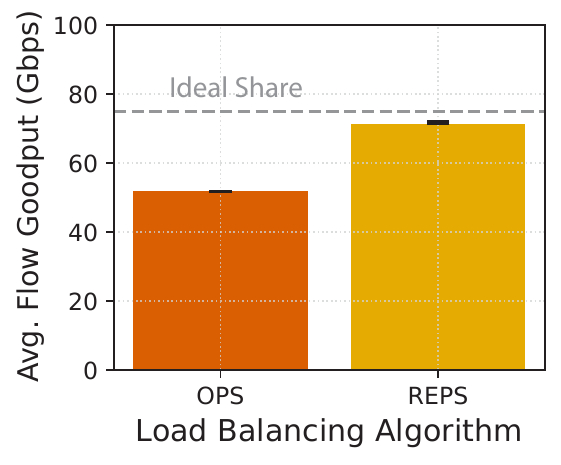}
        \caption{Asymmetric Network}
        \label{fig:reps_symmetric_real:b}
    \end{subfigure}
    \caption{REPS-FPGA impact on goodput.}
    \label{fig:reps_asymmetric_real}
    \vspace{-1em}
\end{figure}

\begin{figure}[t]
    \begin{subfigure}[b]{0.48\linewidth}
        \includegraphics[width=\linewidth]{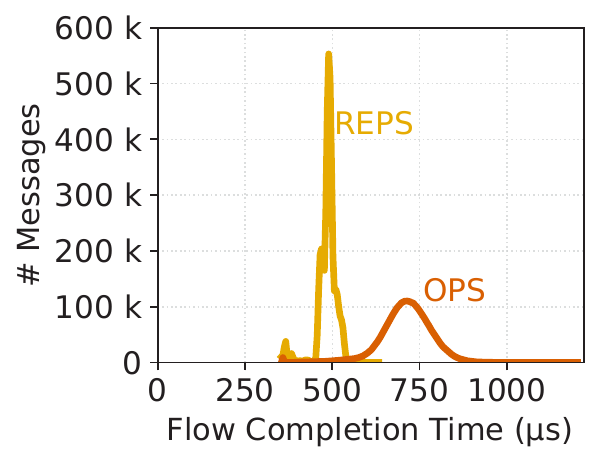}
        \caption{Asymmetric Network}
        \label{fig:reps_asymmetric_real:a}
    \end{subfigure}
    \hfill
    \begin{subfigure}[b]{0.48\linewidth}
        \includegraphics[width=\columnwidth]{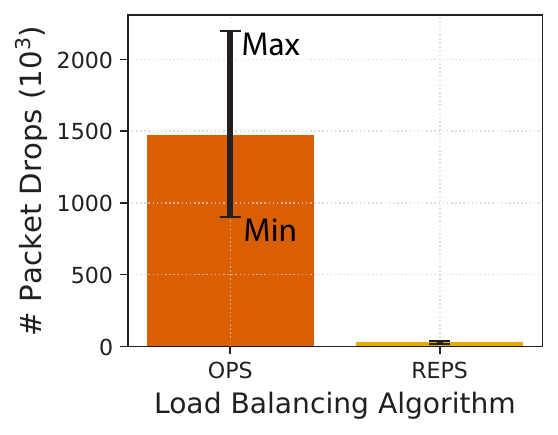}
        \caption{Link Failure}
       \label{fig:reps_link_failure}
    \end{subfigure}
    \caption{REPS-FPGA impact on FCT and packet drops}
    \label{fig:reps_asymmetric_real}
    \vspace{-1em}
\end{figure}

\subsubsection{Network Failures}\label{sec:hw_eval:fail}
We also evaluate the performance of REPS in the presence of network failures. To demonstrate the robustness and resilience of REPS in the context of link failures, Figure~\ref{fig:reps_link_failure} shows total packet drops (average across five runs each) observed in a large-scale 128 endpoint run (endpoints split across 2 T0s connected through 8 T1s) where we abruptly bring down a T0-T1 link during the experiment. While the network is trying to recover from the impact of this event (which in our environment can take in the order of 100s of milliseconds), OPS continues sending packets across all paths (including those affected by the link that went down). The freezing capability of REPS can quickly adapt to such events (within the order of an RTO) and avoid sending packets down the affected paths as the entropy cache is replenished from packets traversing healthy paths. 

\subsection{REPS Applicability}\label{sec:reps_flex}
\changed{In this section, we briefly evaluate REPS, in simulations, under different scenarios by changing the ACK coalescing ratio, the EVS size and, the underlying CC algorithm. Generally, we believe that REPS can work well under many different circumstances, topologies and workloads.}

\subsubsection{ACK Coalescing}\label{sec:ack_ratio}
\changed{We have primarily evaluated REPS without ACK coalescing, as this configuration allows REPS to operate with the most up to date data. However, some transport protocols permit ACK coalescing, where the receiver sends an ACK packet only after receiving $n$ data packets from the sender. In theory, the coalesced ACK packet could return all previous non-ECN marked entropies in its header (we call this configuration \textit{ACK+Carry EVs}). As an alternative approach, each entropy could be assigned a lifespan in the REPS buffer and reused $n$ times (\textit{ACK+Reuse EVs}).}

\begin{figure}[htbp]
    \centering
\includegraphics[width=\linewidth]{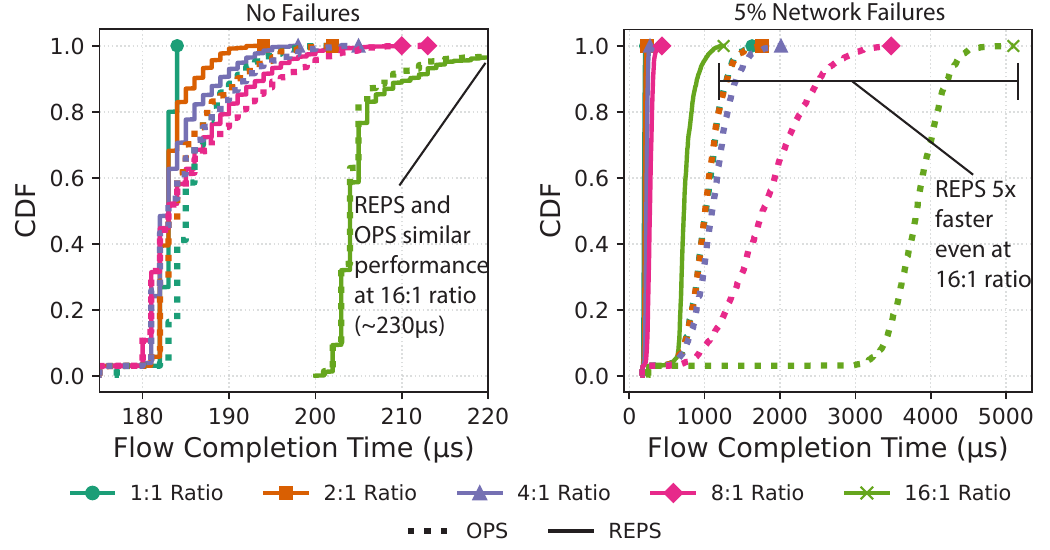}
    \caption{\changed{Performance with different ACK coalescing ratios during an 8MiB permutation.}}
    \label{fig:reps_vs_ops_diff_ack_coalescing}
    \vspace{-0.5em}
\end{figure}

\changed{We run an 8 MiB permutation simulating different ACK coalescing ratios and only the standard REPS configuration. As we can see in the left part of Figure~\ref{fig:reps_vs_ops_diff_ack_coalescing}, REPS with 2:1, 4:1, and 8:1 coalescing ratios does significantly better than OPS, while with 16:1 it starts losing its advantage. However, it should be noted that in case of asymmetries or failures (right Figure~\ref{fig:reps_vs_ops_diff_ack_coalescing}), REPS does much better even at 16:1 ratio. We confirm these results theoretically in Appendix~\ref{appendix:additional_results:ack_theoretical}.}

Figure~\ref{fig:reuse_reps} demonstrates that \textit{Carry EVs} and \textit{Reuse EVs} are the preferred REPS variants for ACK coalescing, particularly under high coalescing ratios. We leave extensions of this approach to future work.

\begin{figure}[htbp]
    \centering
    \includegraphics[width=\linewidth]{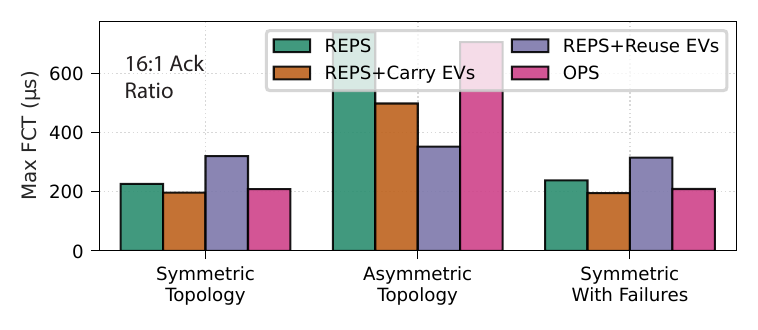}
    \caption{Different REPS variants for ACK coalescing.}
    \label{fig:reuse_reps}
    \vspace{-1em}
\end{figure}

Finally, we note that the trade-off of sending more ACKs is worth the effort if the underlying hardware supports such rate and the impact on the network traffic is minimal (\textasciitilde
1\%).

\subsubsection{EVS Size} \label{sec:evs_size}


Practical implementations make use of hashing functions to map EVs onto output ports. We show, first theoretically and then with simulations, that while using a small EVS can lead to high inherent load balancing issues, using $2^{16}$ EVs is very close to uniformly random.


We investigate the phenomenon that a small EVS leads to poor load balancing of the output ports of a switch in the case of a fat tree topology.

We use a balls-into-bins model~\cite{DBLP:conf/random/RaabS98}, where $m$ balls are thrown, uniformly and independently, at random into $n$ \emph{bins}. The goal is to determine the largest number of balls in any bin, referred to as the \emph{maximum load} $l(m, n)$.  We define the \emph{load imbalance} $\lambda_{m, n}$ as $\frac{l(m, n)}{m/n}-1$, representing the extent to which the most heavily loaded bin exceeds the average. 

In our setting, the output ports correspond to bins and the EVs correspond to balls. 
In our model, $\lambda_{m, n}$ represents the load imbalance of the EVs onto uplinks. Since EVs are chosen for packets uniformly at random, the load balancing of EVs directly affects the load balancing of packets onto output ports.

The load imbalance depends on the average number of balls per bin (i.e., $\frac{m}{n}$): if $m = n$, then the load imbalance is $\Theta\big(\frac{\log n}{\log \log n}\big)$ with high probability. 
However, when the ratio of balls to bins $\frac{m}{n} \gg \log n$, the load imbalance tends to zero with high probability~\cite{DBLP:conf/random/RaabS98}.
In conclusion, for OPS and a fixed number of flows, we expect high load imbalance with a small EVS and near-zero load imbalance as the EVS increases.

We confirm this theoretical analysis with simulations of the load imbalance. Figures \ref{img:one_flow_entropies} and \ref{img:32_flow_entropies} show the distribution of the load imbalance for $1$ and $32$ unique flows, respectively. We note that each flow is from a different sender and will hence have different header fields that will be used in the hashing function regardless of the EV value.
We note that for each case, we throw for each active flow a number of balls equal to the EVS size, with each ball being a unique EV. We can see that for 32 flows, choosing less than $2^{8}$ EVs can lead to more than $10 \%$ load-imbalance, whereas $2^{16}$ EVs guarantee less than a $1 \%$ load imbalance.

\begin{figure}[htpb]
    \begin{subfigure}[b]{0.48\linewidth}
        \includegraphics[width=\linewidth]{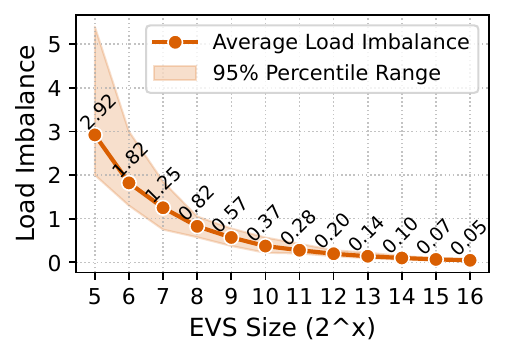}
        \caption{1 flow active}
        \label{img:one_flow_entropies}
    \end{subfigure}
    \hfill
    \begin{subfigure}[b]{0.48\linewidth}
        \includegraphics[width=\columnwidth]{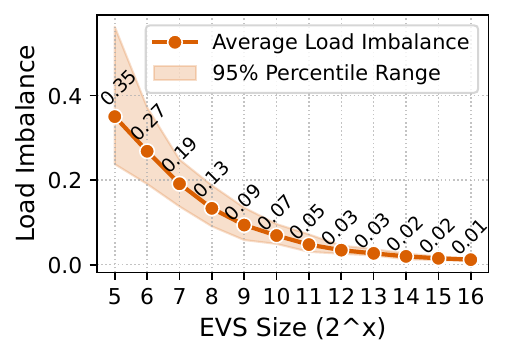}
        \caption{32 flows active}
        \label{img:32_flow_entropies}
    \end{subfigure}
    \caption{Expected load imbalance at a switch (32 uplinks).}
    \vspace{-0.5em}
\end{figure}
\changed{Moreover, this requirement also grows with the number of output ports in a switch as demonstrated in Section~\ref{sec:rec_theory}. On the other hand, REPS, due to its adaptive nature, can drastically reduce the EVS size and still work well. We show this in the left plot of Figure~\ref{fig:reps_vs_ops_features} where we compare, in a real scenario, OPS and REPS when using 32, 256, and 64K EVs. REPS works equally well with 256 and 64K EVs and is only 8\% slower with 32 EVs. On the other hand, OPS is 21\% and 64\% slower with 256 and 32 EVs when compared to 64K. This confirms that REPS could potentially work well even just with 1 byte for the EVS. We note that while OPS could be implemented without using EVS, such as through round-robin selection or by making random choices directly at the switch, these approaches introduce additional challenges  (Section~\ref{sec:rec_theory_lim}).}

\subsubsection{Different CC Algorithms}
REPS has been designed to work with any CC algorithm as long as there is no over-reaction to out-of-order packets and ECN support. In this paper, we have, so far, evaluated REPS working alongside a tuned version of DCTCP (similar to what is used in MPRDMA \cite{10.5555/3307441.3307472}). However, as we will see, REPS can work well even with other algorithms. Moreover, we envision also a version of REPS that could work just with delay if ECN is not supported but we do not go into details here. For example, in the right plot of Figure~\ref{fig:reps_vs_ops_features}, we run an 8 MiB permutation workload without failure for DCTCP, EQDS and a proprietary CC algorithm. REPS can help all of these CC algorithms when compared to OPS.

Intuitively, REPS works on a per-packet basis, meaning each flow simultaneously uses many paths. If one path starts encountering congestion, REPS will quickly route away from it. Temporarily, the CC algorithm might slightly lower the congestion window or sending rate. For instance, if it uses only ECN as its signal and only one path is congested, the decrease would be minimal (assuming a CC that does not overreact to each ECN-mark). Moreover, non-ECN packets from uncongested paths would quickly restore the window. If all paths are congested, REPS cannot add capacity, and the CC rightly lowers the congestion window; this is precisely what we observe in the incast experiments where every balancer behaves the same (Figure~\ref{fig:macro_symm}). In these cases REPS (like all routing algorithms) cannot fix the issue and the CC must kick-in to correctly adjust the sending rate. Finally, we note that one could potentially optimize the CC further for a smoother REPS integration. For example, the CC could activate only after a certain degree of congestion is detected, while leaving the initial work to the load balancer. Having said that, this is mostly orthogonal to REPS and also affects other load balancers.

\begin{figure}[htbp]
    \centering
    \includegraphics[width=\linewidth]{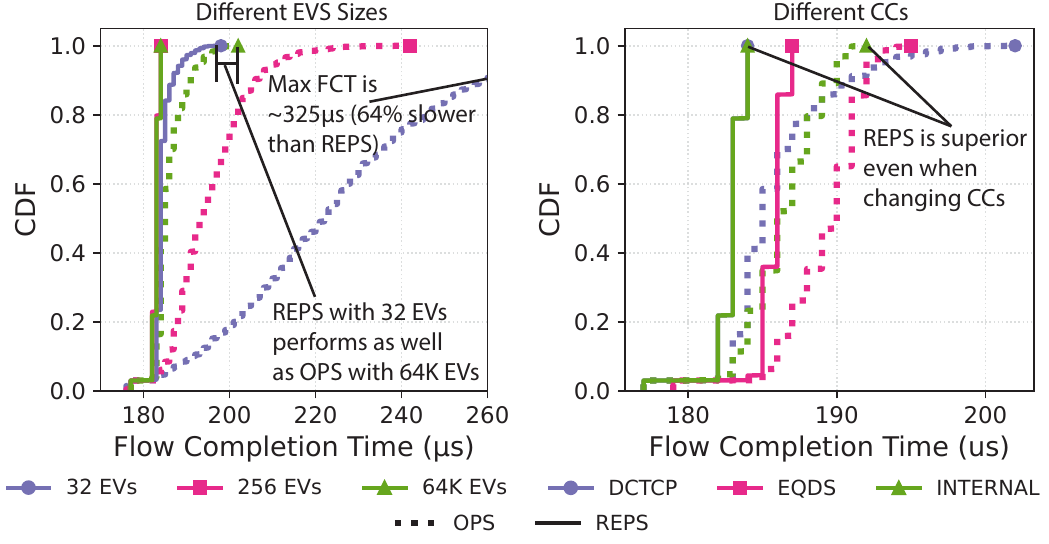}
    \caption{\changed{Performance with different EVs sizes and CCs during an 8MiB permutation.}}
    \label{fig:reps_vs_ops_features}
    \vspace{-1em}
\end{figure}

\subsubsection{Topology Scaling}
We then shift to verifying REPS performance when changing the scale of the topology. To do this, we run a 8MiB tornado workload and vary the number of nodes in the topology. We range from 128 nodes (switches with 16 ports) to 8192 nodes (switches with 128 ports). Moreover, we also vary the EVS size and compare once again REPS to OPS. Figure~\ref{fig:scaling} shows the maximum completion time for each configuration and we run each experiment multiple times to account for randomness in the initial EVs. As we can see, REPS performs well for all topology sizes and almost all EVS sizes with only a small regression when using 16 EVs. This, potentially, means that REPS could work with far less than 16 bits needed for the EV, even at large scale. On the other hand, OPS suffers much more due to its non-adaptivity when having a limited number of EVs with 16 values running more than twice slower than when using the full 16 bits (65536 possible values). We also observe a slight upward trend in the maximum FCT with OPS as topology size increases. The key point is that, with REPS, even if a limited number of EVs prevents each flow from using all available paths, the increased adaptivity compensates for this. In addition, other packet-header fields contribute additional randomness because switch hashing functions combine them with the EV. This observation is consistent with the analysis in Section~\ref{sec:evs_size}.

\begin{figure}[htbp]
    \centering
    \includegraphics[width=\linewidth]{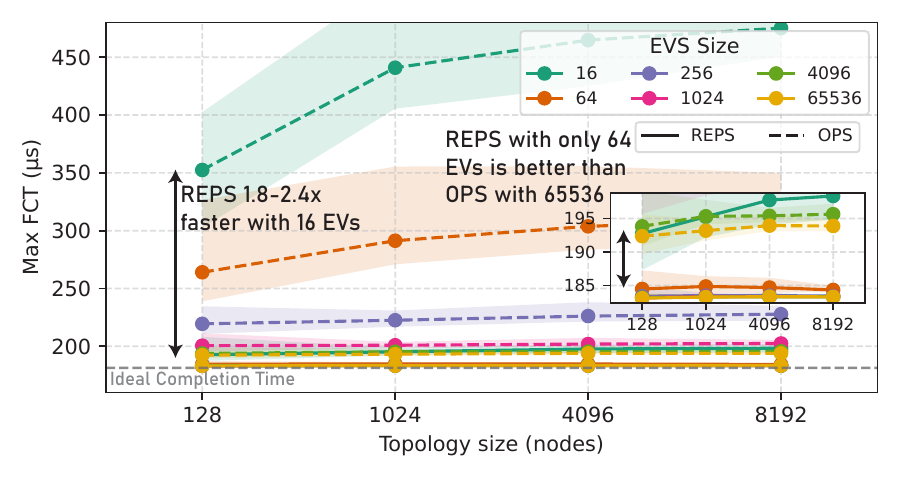}
    \caption{\changed{Performance with different EVS sizes and topology sizes during an 8MiB permutation.}}
    \label{fig:scaling}
    \vspace{-1.5em}
\end{figure}

\section{Theoretical Verification} \label{sec:theo_mot}

\changed{To support our experimental findings, we present a theoretical first principles analysis of OPS and demonstrate how it can lead to arbitrarily long queues. 
To address this issue, we develop a new \textit{recycled balls-into-bins} model and prove its local convergence. This model serves as the theoretical intuition behind the REPS protocol (Section~\ref{sec:reps}).}

\subsection{Recycled Balls-into-Bins Model} \label{sec:rec_theory}
\changed{At maximum injection rate, OPS suffers from severe load-imbalance that eventually leads to queues of arbitrary length building up. We explain this behavior with an infinite batched balls-into-bins model~\cite{DBLP:journals/algorithmica/BerenbrinkFKMNW18, DBLP:journals/dc/BecchettiCNPP19, DBLP:conf/stacs/LosS23}. In contrast, we show that recycling good paths such as in REPS leads to a convergent behavior and logarithmically-bounded queues.} 

\changed{In our switch model, each output port corresponds to a bin. At each time step, every non-empty bin removes one element. Afterward, a new set of balls (packets) arrives and is distributed among the bins. In our setting, we focus on the case where $n$ balls arrive in each time step, representing full throughput.  The maximum queue length at any time step corresponds to the maximum load of any bin at that time step.  Balls are removed in FIFO order from the bins.} 

\changed{In OPS, balls are allocated to bins uniformly at random. In what follows, we assume the EVS is sufficiently large (i.e., 16 bits), allowing us to model the assignments as uniformly random.}
\changed{If balls arrive at a rate of $\lambda n$ for $\lambda < 1$, the process remains stable. The maximum load at any time step is $O\left(\frac{1}{1-\lambda} \log \frac{n}{1-\lambda}\right)$ with high probability~\cite{DBLP:journals/algorithmica/BerenbrinkFKMNW18} and with probability approaching $1$, there is always a bin containing $\Omega\left(\frac{1}{1-\lambda} \log n\right)$ balls~\cite{DBLP:journals/algorithmica/BerenbrinkFKMNW18}. In the limit as $\lambda \rightarrow 1$, this implies that some bin will eventually become arbitrarily overloaded. In the context of load balancing, this means that at the maximum injection rate, \emph{oblivious random spraying leads to unbounded queue lengths}. Intuitively, this occurs because $n$ balls are introduced at each step, but fewer than $n$ may be removed, as some output ports may remain unselected.}

\changed{As $n$ increases, the maximum load grows, exacerbating congestion in OPS. Figure~\ref{fig:ports_study} illustrates this effect for $\lambda = 0.99$, where larger numbers of bins (output ports) result in faster-growing maximum queues.}

\begin{figure}[htbp]
    \centering
    \includegraphics[width=0.88\linewidth]{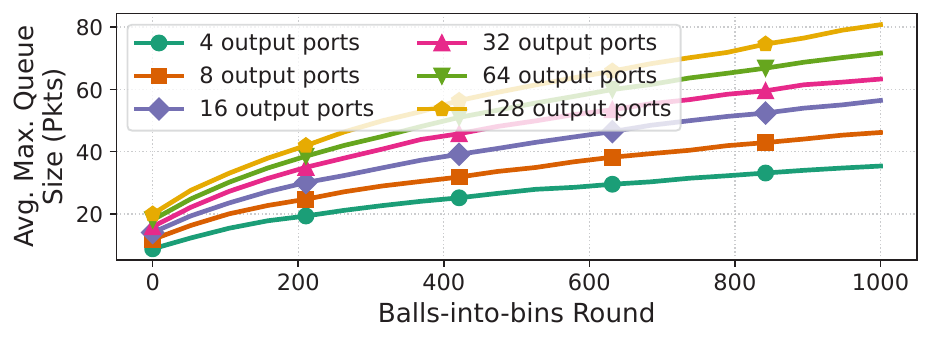}
    \caption{Simulating 1000 rounds of balls-into-bins.} 
    \label{fig:ports_study}
    \vspace{-1em}
\end{figure}


\changed{We propose a new model, called \emph{recycled balls-into-bins} and prove it converges locally to a load-balanced state with bounded queue lengths, even at maximum injection rate.}

\changed{In the recycled balls-into-bins model, we keep a set of $b \cdot n$ colors (for some constant $b$) and a threshold $\tau$. We cycle through the colors in round-robin fashion in batches of $n$ colors. In each time step, we remove a ball from each nonempty bin. If a bin has at most $\tau$ balls, the color of the removed ball remembers the bin, unless it already remembers another bin. If the bin has more than $\tau$ balls, the color forgets its bin.
Then, for each color in the batch that remembers a bin, throw a ball into the remembered bin. For all colors in the batch without a remembered bin, throw a ball uniformly at random.} 
\changed{We show that, for a single switch, \emph{recycled balls-into-bins converges}, meaning all colors remember a bin and keep the same bin remembered.} 
\begin{theorem}\label{thm:convergence}
For $n\geq 16$, $\tau \geq 4 \ln n $, $b \geq 2.4 \ln n$, recycled balls-into-bins converges in $O(n \log n)$ steps. Every bin has $O(\log n)$ elements throughout, with probability $1-o(1)$.
\end{theorem}
\changed{The proof in Appendix \ref{appendix:proof} shows that recycled balls-into-bins converges as bins fill and stabilize below the threshold.}

\begin{figure}[htbp]
    \centering
    \includegraphics[width=0.89\linewidth]{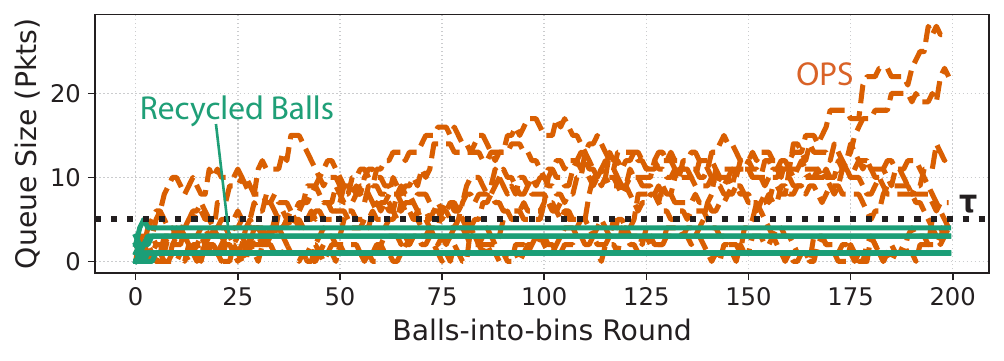}
    \caption{\changed{Simulating 200 rounds of balls-into-bins.}
    }
    \label{fig:ops_recycle}
    \vspace{-1em}
\end{figure}

\changed{To visualize these behaviors, we model the balls-into-bins problem with both OPS and with the recycled balls improvement. In Figure~\ref{fig:ops_recycle}, we set $n = 5$ (for visualization purposes, but this holds true for more realistic $n$ values) and showcase the queues' evolution for OPS and the recycled balls model. The results confirm that while OPS leads to unbounded queue growth, the recycled balls model converges, keeping all queues below the threshold ($\tau$) value, consistent with our previous network simulations (e.g., Figure~\ref{fig:nocc_symm} and Figure~\ref{fig:micro_symm}).}

\subsection{Limitations and Alternatives}  \label{sec:rec_theory_lim}
\changed{While the provided model performs well in an idealized setting with maximal injection rates, in a real network scenario, CC algorithms would reduce the sending rate if a large queue were to build up. However, such activations of the CC would result in a slowdown of the sending rate and, hence, an increase of the completion time of a workload as we can see in Section~\ref{sec:eval:nonos}. Moreover, we note that while real scenarios would have more subtle interactions not modelled here (like using distinct $KMin$ and $KMax$), the key results hold and are consistent with our evaluations in Section~\ref{sec:eval_overall}.
Finally, we note that here we provide a proof only about the local convergence of REPS and that other models such as static path assignments and round-robin across the ports would also work well and potentially provide no queueing. However, such approaches are not viable in reality for a series of reasons: they cannot react well to partial or total failures, they would struggle with multi-tier topologies and they require detailed knowledge of the underlying workloads \cite{facebook}.}

\section{Related Work} 
The literature for load balancing designs is vast, however most of it focuses on TCP-based deployments where out-of-order packets are not desirable~\cite{microTE, hermes, hedera, 201562, 6567015}. In contrast, REPS aligns with a new transport paradigm (UltraEthernet~\cite{ultra}, Adaptive RoCE by NVIDIA~\cite{nvidia}, Falcon~\cite{googleIntroducingFalcon}) where out-of-order packets help to leverage the capacity of multi-path networks.

Generally, most load balancers work at different granularity: per-flow, sub-flow or per-packet. ECMP \cite{ecmp} is a per-flow solution and susceptible to flow collisions and are oblivious to congested paths. Solutions like Hedera \cite{hedera} and MicroTE \cite{microTE} require a global controller which is not desirable in production datacenters \cite{facebook}. For these reasons, many sub-flow solutions like Flowlet Switching \cite{201562}, Flowcut Switching \cite{flowcut}, Presto \cite{flowcell}, CONGA \cite{conga}, PLB \cite{10.1145/3544216.3544226}, FlowBender \cite{flowbender} have been proposed. However, some of these solutions are still congestion oblivious (Flowlet, Presto), react too slowly for AI/ML bursty and intense traffic (PLB, FlowBender), require specialized switches (CONGA). Moreover, most of these solutions are unable to quickly deal with blackholes and failures.

\tom{Load balancers with per-packet granularity, such as OPS \cite{6567015} or DRILL \cite{drill}, help drastically reduce ECMP collisions, but are still oblivious to asymmetries (for OPS) or require switch support (for DRILL). As demonstrated in Section~\ref{sec:theo_mot} OPS can suffer even in symmetric cases. MPRDMA \cite{10.5555/3307441.3307472} also uses ECN for load balancing like REPS, however it requires probing and ACK clocking and does not offer caching of entropies. Hermes \cite{hermes} combines both ECN and delay but it works best with TCP-like protocols and has many parameters with complex tuning.}

\changed{ConWeave \cite{conweave} is designed specifically for RDMA networks and offers a solution by masking out-of-order packets
in commodity RNICs but it requires changes to TOR switches
and has limited scalability.}

Proteus \cite{Proteus} focuses on optimizing the load balancing of lossless PFC networks while REPS focuses on lossy networks.
\section{Conclusion}
\changed{We presented REPS, a simple, lightweight, yet highly effective load-balancing mechanism designed to meet the constraints of next-generation datacenter networks tailored for AI workloads. As demonstrated in our extensive evaluations, conducted through both simulations and FPGA hardware, REPS' adaptive entropy caching enhances end-to-end performance across multiple critical metrics, including average flow completion time, runtime, and packet loss. REPS outperforms ECMP and OPS by up to 6x and 1.25x in symmetric networks, and by up to 5x and 2x in asymmetric networks, outperforming OPS by as much as 100x during short-term transient link failures while reducing packet drops by over 70x. Moreover, REPS consistently beats or matches other state-of-the-art algorithms in all scenarios evaluated. We also showed how REPS can work well under various network configurations demonstrating its flexibility to adapt to different scenarios while remaining light-weight on NIC memory requiring only 25 bytes of state per-connection.}

\begin{acks}
We thank our shepherd, Daniel Amir, and the anonymous reviewers of EuroSys ’26 for their valuable feedback.
This work is supported by the European Union’s Horizon Europe under grant 101175702 (NET4EXA), the Sapienza University Grants ADAGIO and D2QNeT (Bando per la ricerca di Ateneo 2023 and 2024), the European Research Council (ERC) under the European Union’s Horizon 2020 research and innovation program (grant agreement PSAP, No. 101002047), and a CAF America grant. We thank the Swiss National Supercomputing Center (CSCS) for providing the computational resources used in this work. We also thank Marcin Copik, Ales Kubicek, and Nadeen Gebara for providing their valuable feedback.
We used ChatGPT solely for text editing and quality checks.
\end{acks} 

\bibliographystyle{ACM-Reference-Format}
\bibliography{bib}

\newpage
\appendix
\label{appendix}
\section{Freezing Mode in REPS} \label{appendix:freezing_details}
Ideally, REPS should enter freezing mode only upon detecting a network failure. To achieve this, we employ two strategies:

\begin{enumerate}
    \item \textbf{Packet Trimming Support}: When packet trimming is available, distinguishing between packets lost due to congestion and those lost because of network failures becomes more straightforward. We use trimming to identify and separate lost packets with greater accuracy.
    \item \textbf{Absence of Packet Trimming}: In the absence of packet trimming, we analyze the maximum round-trip time (RTT) observed during a period preceding the timeout event. If the maximum RTT immediately before the timeout is high, it indicates that the packet was likely lost due to congestion. Conversely, if the maximum RTT was low, the packet was more likely lost due to a network failure.
\end{enumerate}

In our paper, we focused on scenarios where packet trimming was not supported. However, REPS performs optimally when packet trimming is available, benefiting from both an enhanced loss detection algorithm and a more responsive CC loop.

Regardless of the employed strategy, REPS maintains high performance even if it inadvertently enters freezing mode without an actual network failure. This is because entering freezing mode effectively reduces the EVS size of REPS. As demonstrated in Section~\ref{sec:evs_size}, REPS remains effective with as few as 32 EVs. For instance, we tested REPS with the 16 MiB permutation workload, running it first under normal conditions and then with forced freezing mode activated after $50\;\mu\text{s}$. The results show comparable results for REPS with or without forcing the freezing mode, as illustrated in Figure~\ref{fig:reps_freeze_stable} with only minor instability introduced by activating freezing mode. However, compared to OPS, REPS with the force freezing mode is still significantly more stable and has a finishing time comparable to the standard REPS, both completing the workload slightly faster than OPS.

As an extension of REPS, probing can be incorporated to enhance failure detection precision. However, we chose not to integrate this into REPS at this stage to maintain clarity and focus on the core framework. We also do not discuss fast loss-recovery mechanisms as they are anyway orthogonal to REPS behaviour and could be used to further improve its performance.

\begin{figure}[htbp]
    \centering
\includegraphics[width=\linewidth]{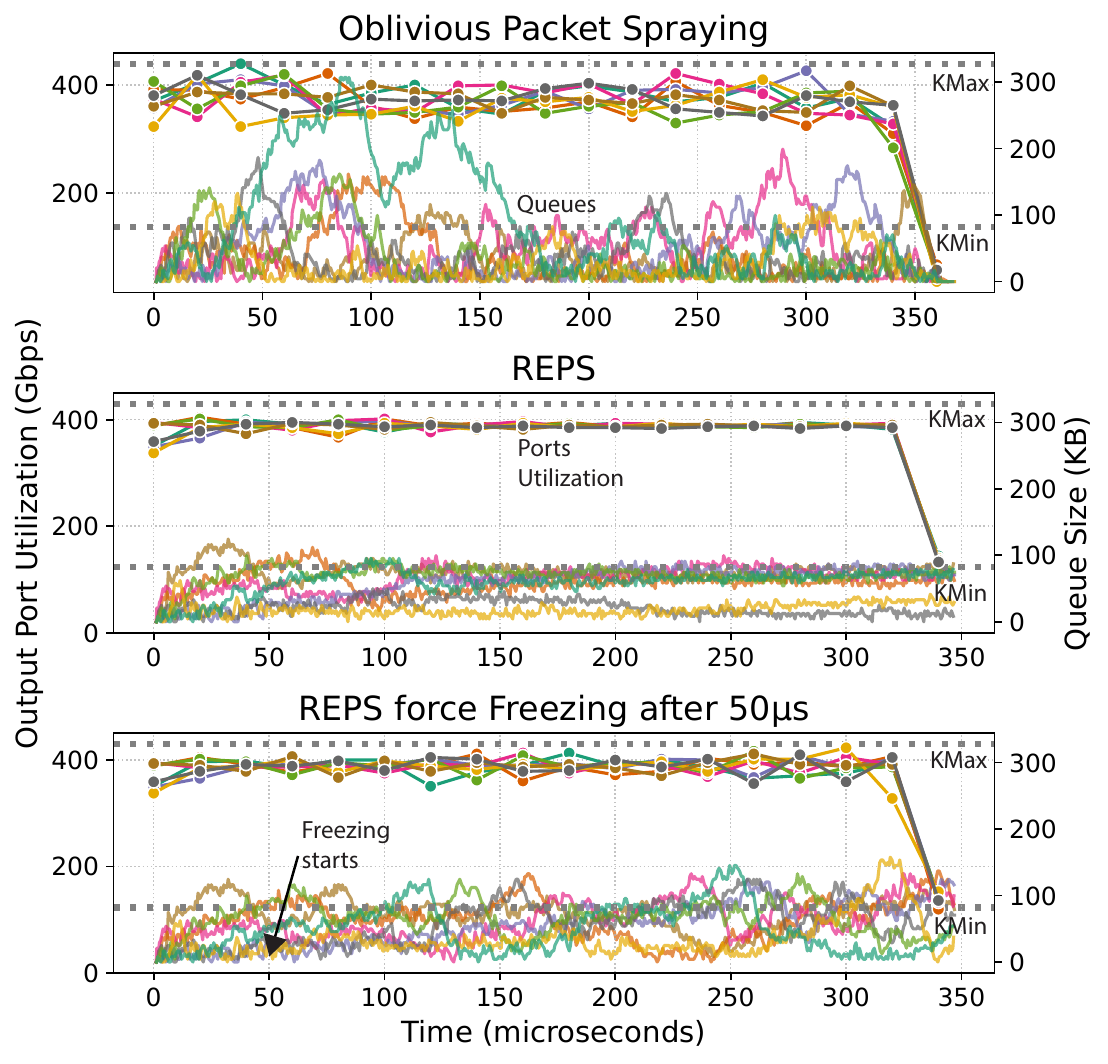}
    \caption{Studying REPS while forcing freezing mode after a fixed period of time.}
    \label{fig:reps_freeze_stable}
\end{figure}

\section{REPS Single Switch Convergence Proof}\label{appendix:proof}
We prove the main theorem that demonstrates the convergence of the recycled balls-into-bins process for a single switch.

The process has two distinct phases. In the first phase, there are still empty bins. The number of balls can still grow in this 'warmup' phase. We will bound the steps it takes to undergo the first phase and bound the maximum load. 
In the second phase, all bins are non-empty. Once a bin is non-empty, it remains non-empty. 
Moreover, once all bins are non-empty, the number of balls in the system remains constant, as for every ball removed a ball is thrown. In this second phase, we will show that the system converges to a state where all bins have at most $\tau$
balls by means of a drift theorem~\cite{DBLP:series/ncs/Lengler20}. Note that once all bins have at most $\tau$ balls and every ball has a distinct color, they will keep having at most $\tau$ balls. Moreover, we will argue about the maximum load and we will show that throughout, as long as the number of colors is at least $b \cdot n \geq (2.4 \ln n)n$, there is at most one ball of any given color with high probability.

Note that much of the proof is devoted to showing convergence of the second phase. If we set the threshold large enough, no bin ever overflows and essentially we only have the first phase, which converges in expected $n \ln n + n$ steps. Since queues are an expensive hardware resource, we are interested in showing good constant factor bounds for the queue lengths rather than giving a weaker asymptotic bound for $\tau$. 

\begin{proof}[Proof of Theorem \ref{thm:convergence}]

\textbf{First Phase.} We say a ball is fresh if it was thrown randomly. 
After $m= 2 n \ln n$ fresh balls are thrown, every bin received at least one fresh ball with probability $1-\frac{1}{n}$, by the coupon collector problem~\cite{DBLP:series/ncs/Doerr20}. 
Either all bins are non-empty before $m$ steps and we are done, or we have thrown at least $m$ fresh balls and all bins are non-empty. Either way, a bin receives at most $O(\log n)$ fresh balls, by balls into bins with $m+n$ balls~\cite{DBLP:conf/random/RaabS98}. Observe that the number of balls in a bin only increases if it receives a fresh ball. Hence, the number of received fresh balls  bounds the number of balls in a bin. 

Since $b > 2.4\ln n$, we have at most one ball per color when the first phase ends.

\textbf{Second Phase}
At the beginning of the second phase, we have at most $2 n \ln n + n$ balls in total and all bins are nonempty.
We define a potential function to measure the drift towards a completely load-balanced state. Let $X^{i}_t$ be the number of balls in bin $i$ at time step $t$ and let $Y^{i}_t = \max(0, X^{i}_t - \tau)$. Our potential is $Y_t=\sum_i Y^{i}_t$.  Clearly, when the potential is $0$, all bins have at most $\tau$ elements. We need to bound the drift $Y_{t+1} - Y_{t}$. 

We first consider the case where each ball has a unique color. The drift depends on the number of bins with more than $\tau$ elements, which we denote by $K_t$. Let us consider the expected drift for one particular bin $i$. The easy cases are:
\begin{align*}
E[Y^{i}_{t+1} - Y^{i}_t | X^{i}_t = x, K_t = k] = 
\begin{cases}
  \frac{k}{n} - 1 & \text{if } x > \tau, \\
  \frac{k}{n} & \text{if } x = \tau \\
  0 & \text{if } x \leq \tau - k, \\
\end{cases}
\end{align*}
If $x$ is in between $\tau-k$ and $\tau$, it gets interesting. The expected drift is still bounded by $\frac{n}{k}$, but this estimate is not strong enough overall. Instead, we need to show that the expected drift is very small for all such bins. 
To this end, we express the expected drift using a binomial random variable $B$ with parameters $\frac{1}{n}$ and $k$ as $E[\max(0, B- (\tau - x))]$. We cannot directly use linearity of expectation because of the maximum function. Instead, using total expectation leads to the following expression (for $z \geq 3$):
\begin{align*}
E[\max(0, B- z)] &= \left(E\left[ B \ | \ B \geq z\right] - z\right) P[B \geq z] \\
&\leq \frac{1}{z - 1} e^{-(z-1)}, 
\end{align*}
where we show the tail bound in Lemma \ref{lem:binomial}. 

We now compute the expected drift overall by considering each type of bin separately. 
By summing over all bins, the total expected drift contribution of bins with less than $\tau - \ln n$ balls is therefore bounded by $\frac{1}{ \ln n }$ (since $z \geq \ln n + 1$). Note that these bins only contribute to the drift if $k > \ln n$.
Next, consider the contribution from bins with at least $\tau - \ln n$ but at most $\tau$ balls. Since we have at most $2n \log n + n$ balls, of which at least $k (\tau + 1)$ balls are in overfull bins, there can be at most $\frac{2n \log n + n - k (\tau + 1)}{\tau - \ln n} \leq \frac{3}{4} n - k$ such bins. They each contribute at most $\frac{k}{n}$, so their total contribution is at most $\frac{3}{4}k - \frac{k^2}{n}$.
Finally, each overfull bin contributes $\frac{k}{n}-1$ to the expectation, for a total contribution of $\frac{k^2}{n} - k$. 
%
%
%
Hence, the total expected drift is 
\begin{align*}
E[Y_{t+1} - Y_t | K_t = k] \leq
\begin{cases}
  \frac{1}{\ln n} - \frac{k}{4}  & \text{if } k > \ln n, \\
  - \frac{1}{4}  & \text{else.} \\
\end{cases}
\end{align*}

Next, we bound the expected drift when there is more than one ball per color. We have two balls of the same color only if a ball does not get removed before the next ball of the same color is thrown. This happens if more than $b$ balls arrive in the same bin in the same time step as the ball (this uses the FiFO-property of the queues). By a tail bound on the binomial distribution, this happens with probability at most $ \frac{ e^{5/4} }{n^{3}}$, by Lemma \ref{lem:binomial-tail} using $b \geq \frac{12}{5}\ln n$. By a union bound, the probability that \emph{any} color in a batch has more than one ball is at most $\frac{ e^{5/4} }{n^{2}}$. The drift is bounded by $n$, so the contribution to the expected drift by this case is at most $\frac{e^{5/4}}{n}$.

Note that for $n \geq 16$,  $\frac{e^{5/4}}{n} + \frac{1}{\ln n} - \frac{1+\ln n}{4} \geq - \frac{1}{8}$ and $\frac{e^{5/4}}{n} - \frac{1}{4} \ge - \frac{1}{32}$. Hence, the  expected drift is at most $-\frac{1}{32}$. 
We conclude by an additive drift theorem~\cite{DBLP:series/ncs/Lengler20,DBLP:journals/nc/HeY04} that the expected time until all bins are below the threshold is at most $E[32 \cdot Y_0] = O(n \ln n)$. If all balls have unique colors at the end, the process converges then. This happens with probability $1-O(\frac{\ln n}{n})$. So with high probability, a single such $Y_t=0$ event suffices and the overall expectation is $O(n \ln n)$. 

Observe that the number of balls in the queues is bounded by $\tau$ plus the maximum load of a batched balls-into-bins process~\cite{DBLP:journals/algorithmica/BerenbrinkFKMNW18} with expected injection rate $\lambda=\frac{k}{n}\leq \frac{1}{2}$. Hence, the number of balls in a bin is $O(\log n)$ throughout with high probability. 

\end{proof}

We use the following tail bound:
\begin{lemma}\label{lem:binomial-tail}
Let ${B}$ be binomially distributed with parameters $\frac{1}{n}$ and $k\leq n$. Then, for any $x \geq 16$:
\[
P[B \geq x] \leq e^{-\frac{5}{4}(x-1)}
\]
\end{lemma}
\begin{proof}
We use an additive Chernoff bound~\cite{im/1175266369}:
\begin{align*}
 P[B \geq E[B] + \delta] & \leq e^{-\frac{\delta^2}{2(E[B] + \delta / 3)}} \\
   & \leq e^{-\frac{\delta}{\frac{2}{\delta} + \frac{2}{3}}} & \\
   & \leq e^{-\frac{5}{4}\delta} & \text{using } \delta \geq 15
\end{align*}
\end{proof}

We use the following lemma to bound the expected drift of bins that are far below the threshold: 
\begin{lemma}\label{lem:binomial}
Let ${B}$ be binomially distributed with parameters $\frac{1}{n}$ and $k\leq n/2$. Then, for any $x \geq 7/2$:
\[
\left(E\left[ B \ | \ B \geq x\right] - x\right) P[B \geq x] \leq \frac{1}{x - 1} e^{-(x-\frac{1}{2})}
\]
\end{lemma}
\begin{proof}

We use a bound from Pelekis~\cite{pelekis2017lowerboundsbinomialpoisson} for $E\left[ B \ | \ B \geq x\right]$. By their theorem on this conditional expectation:
\begin{align*}
E\left[ B \ | \ B \geq x\right]  &\leq x + \frac{(n-x)\frac{1}{n}}{x - \frac{k}{n} + \frac{1}{n}} \leq x + \frac{1}{x - 1} \enspace .
\end{align*}
\noindent
We bound $P[B \geq x]$ similarly as in Lemma \ref{lem:binomial-tail}:
\begin{align*}
 P[B \geq E[B] + \delta] & \leq e^{-\frac{\delta^2}{2(E[B] + \delta / 3)}} \\
   & \leq e^{-\frac{\delta}{\frac{1}{\delta} + \frac{2}{3}}} & \text{using } E[B] \leq \frac{1}{2} \\
   & \leq e^{-\delta} & \text{using } \delta \geq 3
\end{align*}
Because $E[B] \leq \frac{1}{2}$, $P[B \geq x] \leq  e^{-(x-\frac{1}{2})}$.
\end{proof}

\section{Additional Results}\label{appendix:additional_results}

This section presents additional results that could not be included in the main body of the paper.

\subsection{ACK Coalescing Theoretical Modeling}\label{appendix:additional_results:ack_theoretical}
To assess REPS performance under different ACK coalescing ratios, we validate it using the theoretical model from Section~\ref{sec:rec_theory}.

\begin{figure}[htbp]
    \centering
    \includegraphics[width=\linewidth]{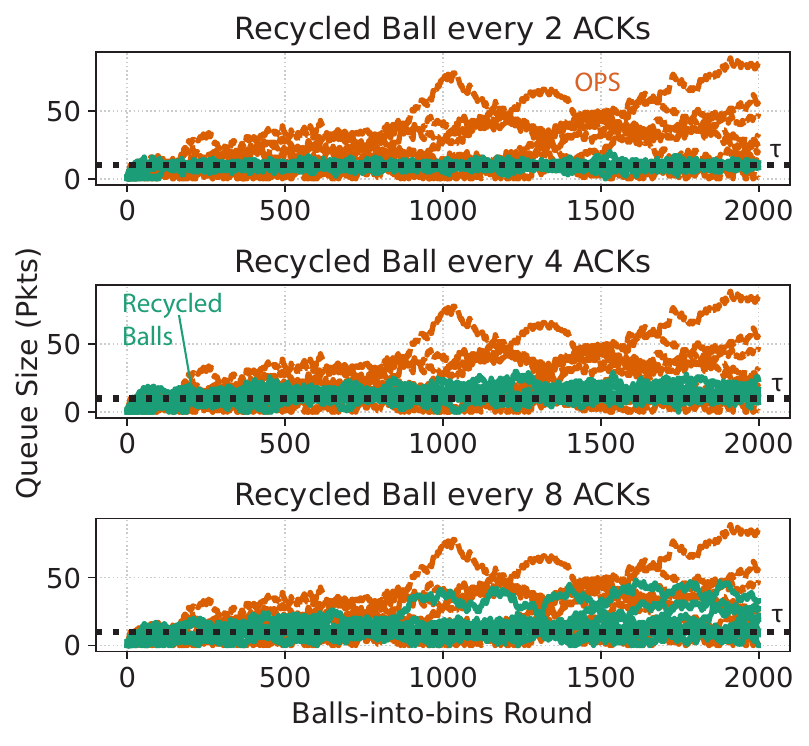}
    \caption{\changed{Performance with different ACK coalescing ratios using the balls-into-bins models.}}
    \label{fig:recycled_balls_compression}
    \vspace{-0.5em}
\end{figure}

Figure~\ref{fig:recycled_balls_compression} shows how the recycled balls model performs well, even with less frequent recycling. While we lose the guarantee of consistently staying below $\tau$, with 2:1 and 4:1 recycling ratios, the queues barely exceed this threshold. An 8:1 ratio still proves slightly more advantageous than OPS.

\subsection{Different Tiers} \label{appendix:additional_results:tiers}

We aim to verify that REPS performs effectively with fat-tree topologies that have three tiers. This scenario poses a slightly greater challenge for REPS, as a single EV must manage two hops. Nonetheless, there is no intrinsic reason why REPS should not perform well in such a topology. 

To validate this, we execute the synthetic benchmark using the symmetric topology described in Section~\ref{sec:eval:nonos}, but with three tiers instead of two. The results, shown in Figure~\ref{fig:3tiers}, indicate that REPS performs comparably to the two-tier topology.

\begin{figure}[htbp]
    \centering
\includegraphics[width=\linewidth]{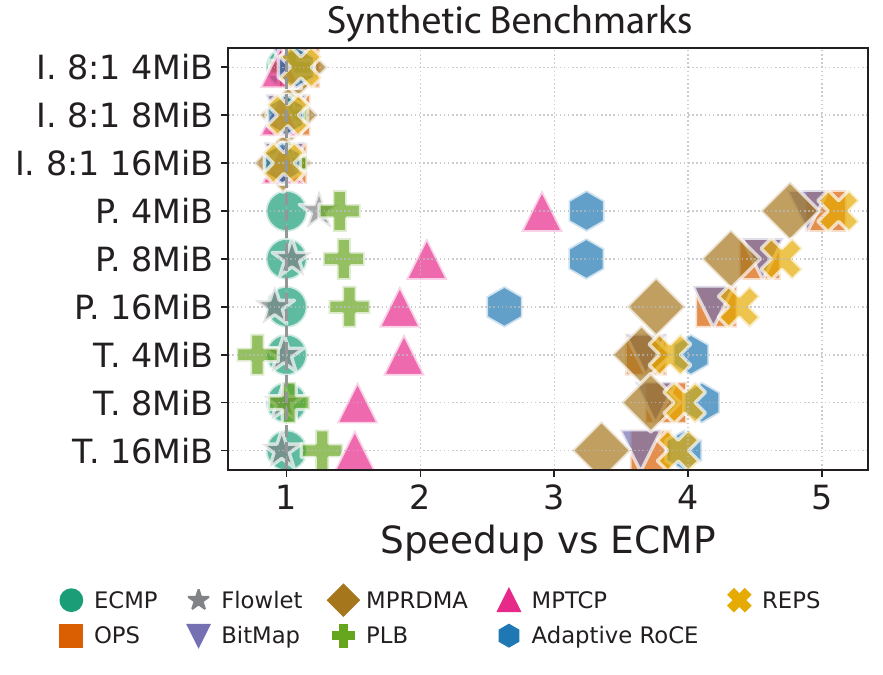}
    \caption{Studying REPS with a 3 tiers fat tree topology.}
    \label{fig:3tiers}
    \vspace{-0.5em}
\end{figure}

\subsection{Incremental Failures} \label{appendix:additional_results:incr_fail}
To further demonstrate the resilience of REPS under failures, we conducted an experiment where we incrementally failed all but one of a switch's uplinks at 200 µs intervals. Figure~\ref{fig:incremental_fail} presents the permutation from the perspective of the failing switch, where three of the four uplinks were permanently disabled in a staggered manner. As expected, REPS enters freezing mode immediately after the first failure, ensuring that failing output ports are avoided. Notably, small utilization spikes are observed on the failing links when REPS exits freezing mode to verify if the failure has been resolved. Since the failures are permanent for the duration of the experiment, REPS promptly re-enters freezing mode after detecting unresolved issues.

In this drastic scenario, OPS performs 40× worse than REPS, primarily due to its inability to avoid broken links, resulting in numerous retransmissions and reduced congestion windows.

\begin{figure}[t]
    \centering
    \includegraphics[width=\linewidth]{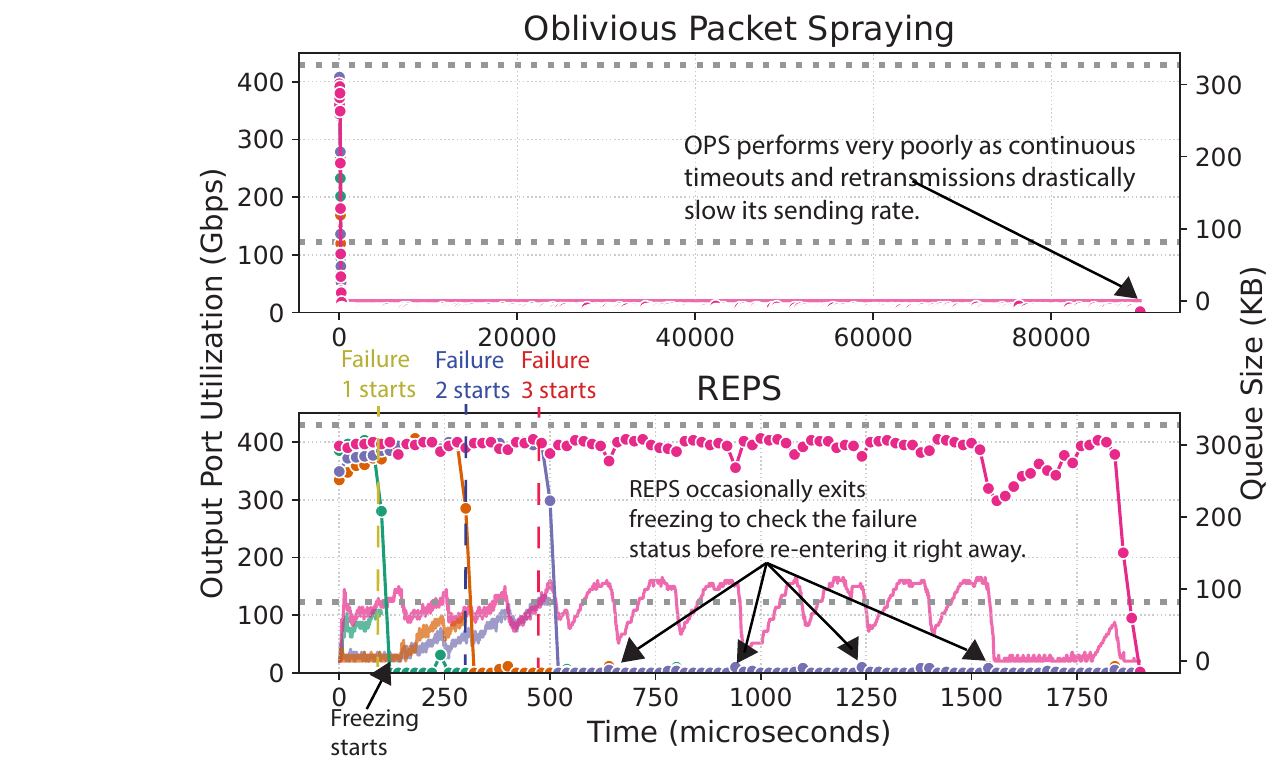}
    \caption{\changed{REPS vs. OPS in a 32 MiB permutation with incremental persistent failures.}}
    \label{fig:incremental_fail}
    \vspace{-0.5em}
\end{figure}

\subsection{Impact of Freezing Mode} \label{appendix:additional_results:freezing}
To demonstrate the impact of having freezing mode in REPS, we evaluate three different scenarios with and without freezing mode. In Figure~\ref{fig:freezing_impact}, we show the effect of having freezing mode available versus disabling it. As expected, in scenarios without failures, REPS performs the same with or without freezing mode. However, when adding failures in the form of 1\% cables failure, we can see freezing mode helping, with almost a 25\% gain over not having it. However, we highlight that REPS remains competitive even when running without freezing mode. This could be used to simplify deployments if needed.

\begin{figure}[htbp]
    \centering
    \includegraphics[width=\linewidth]{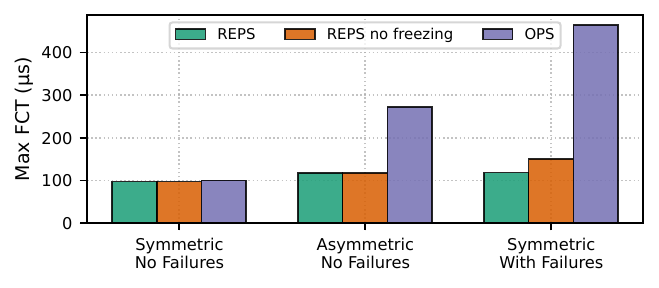}
    \caption{Impact on runtime of disabling freezing mode in REPS.}
    \label{fig:freezing_impact}
    \vspace{-0.5em}
\end{figure}

\section{Additional Data}\label{appendix:additional_data}
The datacenter traces used throughout this paper have been previously used in a number of similar works \cite{201562, 10.1145/1851275.1851192}. In particular we use traces provided by Facebook and a series of traces used for web search in production clusters. 
The CDF distribution for such traces can be seen in Figure~\ref{fig:distribution}. For most of the paper we focus exclusively on the WebSearch traces.
\begin{figure}[htbp]
    \centering
\includegraphics[width=\linewidth]{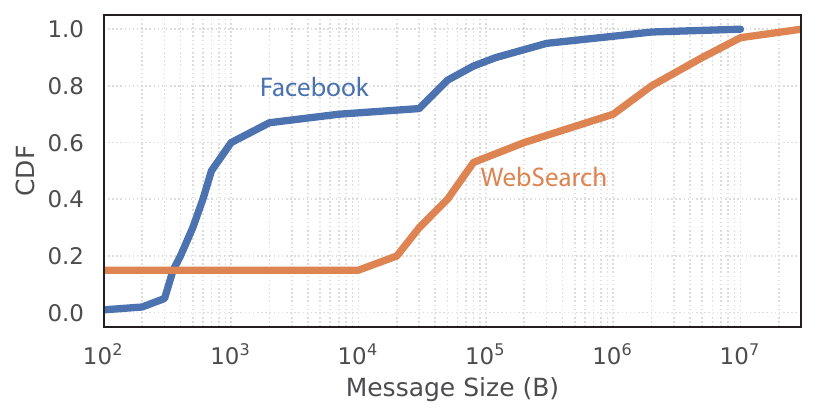}
    \caption{CDF for different data center traces}
    \label{fig:distribution}
    \vspace{-0.5em}
\end{figure}

\clearpage
\section{Artifact Appendix}

\subsection{Abstract}
This appendix provides concise instructions for installing and executing the REPS artifact. Each figure is generated by a dedicated Python script to ensure reproducibility. We also include Bash scripts to run the full experiment suite or any user-selected subset. Note that the AD was prepared earlier than the camera-ready paper, so some figure numbers may mismatch the camera-ready version.

\subsection{Description \& Requirements}

\subsubsection{How to access}
We provide a link to a GitHub repository and a link to a Zenodo DOI:
\begin{itemize}
    \item GitHub: \url{https://github.com/tommasobo/REPS_EuroSys_Artifact}
    \item Zenodo: \url{https://doi.org/10.5281/zenodo.17054005}
\end{itemize}

In the content of the repository we also offer instructions inside the \texttt{README}.

\subsubsection{Hardware dependencies}
No specialized hardware is required. All experiments were performed on a commodity x86~CPU. Any similar processor capable of running the simulator should suffice to reproduce the results. For our runs we used 40GB of RAM, but we expect the simulations to run also with 32GB of RAM. 16GB of RAM might also work but might require lowering the number of parallel simulations.

\subsubsection{Software dependencies}
The artifact requires \texttt{C++17} and \texttt{Python3.8} as the minimum dependencies for building and running both the \texttt{htsim} simulator and the relative Python scripts. All experiments and computational artifacts in this paper were run locally on Ubuntu~22.04~LTS using WSL~2.

To process and visualize the raw data, we rely on a set of standard Python packages, \textit{i.e.,} \texttt{seaborn}, \texttt{scipy}, \texttt{numpy}, and \texttt{pandas}. We do provide a \texttt{requirements.txt} that users can use to install all needed Python dependencies.

\subsubsection{Benchmarks} 
All the benchmarks are self contained and do not depend on external data.

\subsection{Set-up}
After obtaining the source code, either via \texttt{git clone} or by downloading and extracting the Zenodo archive, install the required packages by running:
\begin{itemize}
\item \texttt{./reps\_pkg\_install.sh}
\end{itemize}

If you encounter errors or want to avoid conflicts with local packages, run these commands inside a clean Python environment created with \texttt{venv}. For example:
\begin{itemize}
\item \texttt{python3 -m venv .venv}
\item \texttt{source .venv/bin/activate}
\end{itemize}

Then, from the \texttt{./htsim/sim} directory, install \texttt{htsim} with:
\begin{itemize}
\item
\begin{verbatim}
make clean && cd datacenter/ && make clean && 
cd .. && make -j 8 && cd datacenter/ && make -j 8 && 
cd ..
\end{verbatim}
\end{itemize}

\subsection{Evaluation workflow}

\subsubsection{Major Claims}
Our main claim of the paper is that:

\begin{itemize}
    \item \textit{(C1): REPS consistently matches or exceeds the performance of modern load-balancing schemes. On symmetric networks, it always slightly outperforms OPS; under failures or asymmetric conditions, it delivers clear gains over OPS and other adaptive approaches.}
\end{itemize}

\subsubsection{Experiments}
Please note that minor discrepancies from the submitted version may occur due to hardware or software differences and inherent randomness, despite our efforts to fix all random seeds. Additionally, some plots were manually annotated and enhanced during post-processing. However, for most figures we expect an exact match. With that in mind, running and analyzing the code is straightforward and explained in the next paragraph.

\textit{[Execution]}
We provide in the \texttt{/artifact\_scripts} folder, a number of Python scripts that will run the code needed for their respective figures. For example, to generate Figure 1, one must run the \texttt{fig\_1\_symmetric\_micro.py} Python script.

To facilitate the process, we provide three Bash scripts that can run most or all experiments needed to reproduce the results of the paper. In particular:
\begin{itemize}
  \item \texttt{reps\_quick.sh}: runs quickly (less than 2 hours) and generates a small subset of Figures.
  \item \texttt{reps\_medium.sh}: runs in about 6 hours and generates almost all Figures of the paper.
  \item \texttt{reps\_full.sh}: runs all experiments in the paper and may take up to a day.
\end{itemize}

Exact running times depend on hardware performance.

~\\
\textit{[Results]}
Once experiments are run, results are automatically generated in the \texttt{artifact\_results/} directory.
Within this directory, each experiment has its own subfolder containing the plots and, when applicable, the \texttt{raw\_data/}. 

For instance, to visualize the first Figure, one must navigate to \texttt{artifact\_results/fig\_1\_symmetric\_micro/plots} and open its contents. In most cases we provide the plots in both \texttt{.pdf} and \texttt{.png} format.

\end{document}